\documentclass[a4paper,12pt]{amsart}
\usepackage[utf8]{inputenc}
\usepackage{geometry}
\usepackage{amsthm}
\usepackage{hyperref}
\usepackage{microtype}
\savegeometry{paper}
\usepackage{graphicx,geometry,color,amsmath,amsfonts,amssymb,amscd,amsthm,amsbsy,upref,comment,pgfplots,tikz}
\usepackage{pstricks}
\usepackage{appendix}
\usepackage{dsfont}
\usepackage{amsmath}
\usepackage{setspace}
\usepackage{amssymb}
\usepackage{amsthm}
\usepackage{mathtools}
\usepackage{comment}
\usepackage[shortlabels]{enumitem}
\usepackage{tikz,lipsum,lmodern,syntonly,caption,subcaption}
\usepackage[most]{tcolorbox}
\usepackage[backend=biber,style=numeric,bibencoding=utf8,isbn=false,url=false,sorting=nty,doi=false,giveninits=true,maxbibnames=10,maxalphanames=10]{biblatex}  
\usepackage{csquotes}

\usepackage{xcolor}
\usepackage{blindtext}
\AtBeginBibliography{\small}
\renewbibmacro{in:}{}  

\definecolor{beaver}{rgb}{0.62, 0.51, 0.44}

\allowdisplaybreaks

\numberwithin{equation}{section}
\newtheorem{theorem}{Theorem}[section]
\newtheorem{corollary}[theorem]{Corollary}

\newtheorem{proposition}[theorem]{Proposition}

\newtheorem{lemma}[theorem]{Lemma}

\theoremstyle{definition}
\newtheorem{remark}[theorem]{Remark}
\newtheorem{definition}[theorem]{Definition}

\newcommand{\Z}{\mathbf Z}

\newcommand{\R}{\mathbf R}
\newcommand{\C}{\mathbf C}

\newcommand{\E}{{\mathbb E}}

\newcommand{\F}{\mathcal{F}}
\renewcommand{\S}{\mathbf{S}}
\newcommand{\B}{\mathbf{B}}

\newcommand{\vp}{\varepsilon}

\makeatletter
\newcommand{\distas}[1]{\mathbin{\overset{#1}{\kern\z@\sim}}}%
\newsavebox{\mybox}\newsavebox{\mysim}
\newcommand{\distras}[1]{%
  \savebox{\mybox}{\hbox{\kern3pt$\scriptstyle#1$\kern3pt}}%
  \savebox{\mysim}{\hbox{$\sim$}}%
  \mathbin{\overset{#1}{\kern\z@\resizebox{\wd\mybox}{\ht\mysim}{$\sim$}}}%
}
\makeatother

\usepackage{breqn}

\title[Sharp Declipping and Unlimited Sampling]{On Sharp stable recovery from clipped \\ and folded measurements}

\thanks{The first author was supported by NSF and Simons Research Collaborations on the Mathematical and Scientific Foundations of Deep Learning. The second author was supported by NSF award DMS 2154931. }

\begin{document}
\author{Pedro Abdalla}
\address{Department of Mathematics\\ University of California, Irvine, USA}
\email{pabdalla@uci.edu}
\author{Daniel Freeman}
\address{Department of Mathematics and Statistics\\
St Louis University\\
St Louis, MO   USA} \email{daniel.freeman@slu.edu}
\author{Jo\~ao P. G. Ramos}
\address{ IMPA\\  Rio de Janeiro, Brazil}
\email{joaopgramos95@gmail.com}
\author{Mitchell A. Taylor}
\address{Department of Mathematics\\
ETH Z\"urich, Switzerland
} \email{mitchell.taylor@math.ethz.ch}

\begin{abstract}

We investigate the stability of vector recovery from random linear measurements which have been either clipped or folded. This is motivated by applications where measurement devices detect inputs outside of their effective range.
As examples of our main results,  we prove  \emph{sharp} lower bounds on the recovery constant for both the declipping and unfolding problems  whenever samples are taken according to a uniform distribution on the sphere. Moreover, we show such estimates under (almost) the best possible conditions on both the number of samples and the distribution of the data. 
We then prove that all of the above results have suitable \emph{(effectively) sparse} counterparts. In the special case that one restricts the stability analysis to vectors which belong to the unit sphere of $\R^n$, we  show that the problem of declipping directly extends the \emph{one-bit compressed sensing} results of Oymak-Recht and Plan-Vershynin. 
\end{abstract}
\maketitle

\section{Introduction}

This manuscript deals with different problems of \emph{recovery} type from incomplete data. As a first example of a problem of such kind, we mention the classical problem of \emph{phase retrieval}: given sampled information on the absolute values of a signal, we wish to recover the original signal modulo multiplication by a constant factor of the form $e^{i \alpha}$, for some $\alpha \in \R.$ This problem, and related versions thereof, first arose from fundamental questions in physics, such as the classical Pauli problem, which asks whether recovery of a particular function should be possible from its absolute value together with the absolute value of its Fourier transform. In more recent times, however, phase retrieval problems have become of pivotal importance for fields such as signal processing, in particular due to numerous applications in crystallography, microscopy and imaging \cite{Drenth,allman,corbett,hurt,ismagilov,klibanov,millane,Rosenblatt,seelamantula}.  This has motivated significant research into the mathematics of phase retrieval and  we refer the reader to \cite{alaifari-daubechies,balan,bandeira-pr,  Christ-Mitch-Ben, grohs-liehr,jaming-liehr} for a sample of mathematical contributions to the subject.  

Applications of phase retrieval require being able to recover a vector $u\in\C^n$ from phaseless measurements  of the form $(|\langle X_j,u\rangle| )_{j=1}^m$.  It is thus important to understand both what measurement vectors $(X_j)_{j=1}^m$ in $\C^n$ allow this to be possible and how stable is the recovery.  Unfortunately, there are no known deterministic constructions of such $(X_j)_{j=1}^m$ which allow for both dimension independent stability and the choice of $m$ on the order of $n$. However, under mild assumptions on the distributions, a collection of random vectors $(X_j)_{j=1}^m$ in $\C^n$ (or $\R^n$) will with high probability allow for the uniformly stable recovery of any $u\in \C^n $ (or $\R^n$) from the measurements $(|\langle X_j,u\rangle| )_{j=1}^m$ when $m$ is on the order of the dimension $n$ \cite{candes-eldar, candes-xiaodong,eldarmendelson,complex-subgaussian-pr,Krahmer-Liu}.  Our goal is to prove corresponding stable recovery guarantees, except instead of recovering $u$ from random phaseless linear measurements, we consider  the recovery of $u$ from random clipped or folded  linear measurements instead. 
 
 Measurement devices in applications with a finite dynamic range cannot distinguish very large measurements from some fixed maximum and likewise cannot distinguish very small measurements from some fixed minimum.  In situations where a device receives an input outside of its range, it is typical for the device to simply output the maximum or minimum value, depending on whether the input was too large or too small.  When one is working with a continuous signal over time, the output results in a graph which is clipped at the top and the bottom.  Recovering the true signal from clipped measurements is a common and significant problem in many applications in audio processing and electrical engineering.  This problem becomes even more difficult when one is working with discrete samples as it is no longer possible to use local information around the clipped portions of the signal to recover the true values.
The following shows the graph of a continuous signal $(\langle x,x_t\rangle)_{t\in[0,1]}$, the clipped continuous signal $(\Phi_\lambda(\langle x,x_t\rangle))_{t\in[0,1]}$, and a random discrete sampling $(\Phi_\lambda(\langle x,x_{t_j}\rangle))_{j=1}^m$.

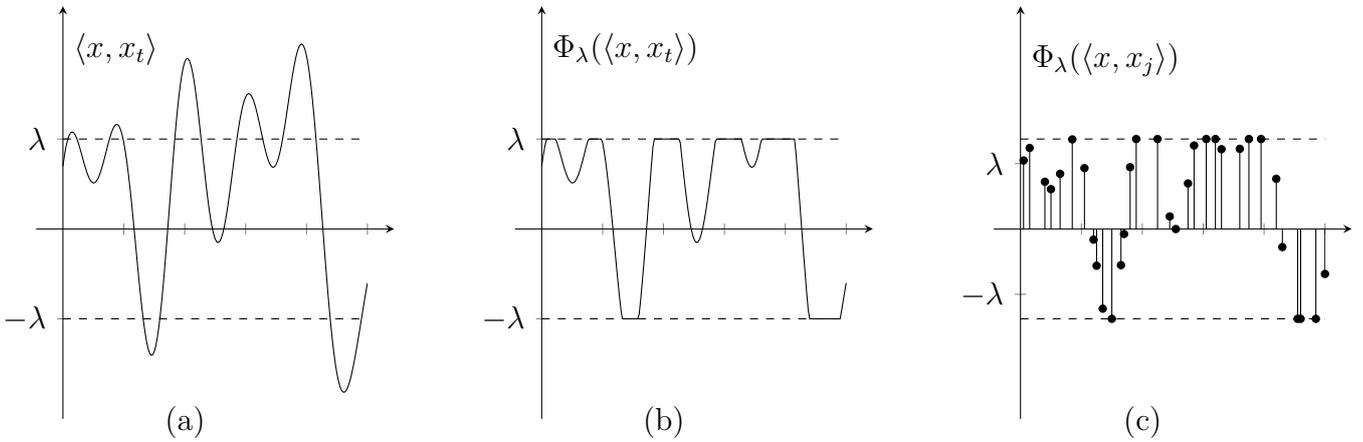
\begin{figure}[htp]
\begin{subfigure}[t]{0.3\textwidth}
\centering
\begin{tikzpicture}
\begin{axis}[width=2.2in, height=2.5in, domain=0:30, xmin=0, xmax=1,
xtick={}, xticklabels={$0$,$1$}, extra x ticks={}, extra x tick labels={}, extra x tick style={xticklabel style={xshift=-0.9ex}}, 
ytick={-.55,0,.55}, yticklabels={$-\lambda$,0,$\lambda$}, ymin=-1, ymax=1.2
,legend pos=outer north east,legend cell align=left,axis lines=center,axis line style={shorten >=-10pt, shorten <=-10pt}
,x label style={at={(current axis.right of origin)},anchor=north, below=3mm,right=5mm}
]

\addplot[domain=0:1,samples=201]{1.2*(5.5*x^2*(1-x)^2 +0.42*sin(2*3*pi*deg(x))+0.3*sin(2*2*pi*deg(x))+0.42*cos(2*5*pi*deg(x))-.1+.25*x+.42*e^(-100*(x-.9)^2)-1*e^(-150*(x-1)^2))};
\addplot[dashed,domain=0:1,samples=3]{0.55};
\addplot[dashed,domain=0:1,samples=3]{-0.55};

\node[scale=1,anchor=west] at (axis cs: 0,1.1) {$\langle x, x_{t}\rangle$};

\end{axis}
\end{tikzpicture}    

(a) 
\end{subfigure}
\hfill
\begin{subfigure}[t]{0.3\textwidth}
\centering
\begin{tikzpicture}
\begin{axis}[width=2.2in, height=2.5in, domain=0:30, xmin=0, xmax=1, xtick={}, xticklabels={$0$,$1$}, extra x ticks={}, extra x tick labels={}, extra x tick style={xticklabel style={xshift=-0.9ex}}, ytick={-.55,0,.55}, yticklabels={$-\lambda$,0,$\lambda$}, ymin=-1, ymax=1.2
,legend pos=outer north east,legend cell align=left,axis lines=center,axis line style={shorten >=-10pt, shorten <=-10pt}
,x label style={at={(current axis.right of origin)},anchor=north, below=3mm,right=5mm}
]

\addplot[domain=0:1,samples=201]{max(-.55,min(.55,1.2*(5.5*x^2*(1-x)^2 +0.42*sin(2*3*pi*deg(x))+0.3*sin(2*2*pi*deg(x))+0.42*cos(2*5*pi*deg(x))-.1+.25*x+.42*e^(-100*(x-.9)^2)-1*e^(-150*(x-1)^2))))};
\addplot[dashed,domain=0:1,samples=3]{0.55};
\addplot[dashed,domain=0:1,samples=3]{-0.55};

\node[scale=1,anchor=west] at (axis cs: 0,1.1) {$\Phi_\lambda(\langle x, x_{t}\rangle)$};

\end{axis}
\end{tikzpicture}    

(b) 
\end{subfigure}
\hfill
\begin{subfigure}[t]{0.3\textwidth}
\centering
\begin{tikzpicture}
\begin{axis}[width=2.2in, height=2.5in, domain=0:1, xmin=0, xmax=1,  xtick={}, xticklabels={$0$,$1$}, extra x ticks={}, extra x tick labels={}, extra x tick style={xticklabel style={xshift=-0.9ex}}, ytick={-0.4,0.4}, yticklabels={$-\lambda$,$\lambda$}, ymin=-1, ymax=1.2,legend pos=outer north east,legend cell align=left,axis lines=center,axis line style={shorten >=-10pt, shorten <=-10pt}
,x label style={at={(current axis.right of origin)},anchor=north, below=3mm,right=5mm},ylabel={$\Phi_{\lambda}(\langle x, x_{j}\rangle)$}]

\addplot+[ycomb,color=black,mark size=1.375pt,mark options={fill=black}] coordinates 
{ (	0.01	,	0.41889072503926	)
(	0.03	,	0.495443371659421	)
(	0.08	,	0.288434401953001	)
(	0.1	,	0.243220691732511	)
(	0.13	,	0.337180914260092	)
(	0.17	,	0.547696951411132	)
(	0.21	,	0.371953561557486	)
(	0.24	,	-0.06559312159928	)
(	0.25	,	-0.22546875	)
(	0.33	,	-0.221397815815752	)
(	0.38	,	0.55	)

(	0.49	,	0.0763265322388639	)
(	0.51	,	-0.000873748757645083	)
(	0.55	,	0.278341697799505	)
(	0.61	,	0.55	)
(	0.64	,	0.55	)
(	0.66	,	0.487689833409819	)
(	0.72	,	0.489949877995612	)
(	0.79	,	0.55	)
(	0.84	,	0.306447695405408	)
(	0.92	,	-0.55	)
(	0.97	,	-0.55	)
(	1	,	-0.275490634707995	)
(	0.27	,	-0.48808339018536	)
(	0.36	,	0.377792961541682	)
(	0.45	,	0.55	)
(	0.75	,	0.55	)
(	0.86	,	-0.111100537351346	)
(	0.91	,	-0.55	)
(	0.3	,	-0.55	)
(	0.34	,	-0.031249919280266	)
(	0.57	,	0.510215147492429	)
};

\addplot[dashed,domain=0:60,samples=3]{0.55};
\addplot[dashed,domain=0:60,samples=3]{-0.55};

\end{axis}
\end{tikzpicture}

(c)

\end{subfigure}

\caption{Effect of clipping and sampling: (a) Unclipped signal (b) Clipped signal (c) Clipped and randomly sampled signal.} \label{fig:saturation-continuous}
\end{figure}

Our first main goal is to study the following array of problems: let $\lambda>0$ be a given constant, and let $m$ independent copies $X_1,\ldots,X_m$ of a random vector $X$ be given.

\begin{enumerate}
    \item[(I)]What is the smallest value of $m=m(\lambda,n)$ so that for all $u,v\in \B_{\R^n}$ the following holds with high-probability:
\begin{equation}\label{declip def}
    \frac{1}{m}\sum_{i=1}^m \left|\Phi_{\lambda}(\langle X_i,u\rangle) - \Phi_{\lambda}(\langle X_i,v\rangle)\right|^2 \ge c\|u-v\|_{2}^2,
\end{equation}
for some constant $c=c(\lambda)$ depending only on $\lambda$?

\vspace{2mm}

    \item[(II)] Assuming that we know certain information on the random vector $X$ (e.g., that it is uniformly distributed on the Euclidean sphere $\sqrt{n}\S^{n-1}$) what is the best possible value of the constant $c(\lambda)$ that we may insert in the lower bound \eqref{declip def} above?

\vspace{2mm}

    \item[(III)] What conditions on the distribution of the random vector $X$ guarantee that for each $\lambda > 0$ there exists $\tilde{c}(\lambda)>0$ such that whenever $m \ge \tilde{c}(\lambda)n$, then \eqref{declip def} holds? What is the sharp dependence of $\tilde{c}(\lambda)$ on $\lambda$ and what are necessary and sufficient conditions on $X$ such that the best constant found in (II) may be achieved? 
\end{enumerate}  

In the above problems, $\Phi_\lambda$ is the ``clipping operator" at the threshold $\lambda$, which acts on vectors component-wise via 
$$\Phi_\lambda(x)= -\lambda\vee x\wedge \lambda.$$
In other words, $\Phi_\lambda$ acts as the identity operator when $x\in [-\lambda,\lambda]$ and records the value $\lambda$ if $x>\lambda$ and $-\lambda$ if $x<-\lambda$.  A frame theoretic approach to saturation recovery is given in \cite{alharbi2024declipping,Freeman-Haider} which characterizes when a given frame $(X_i)_{i=1}^m$ for $\R^n$ allows for the stable recovery of every $u\in B_{\R^n}$ from the values  $(\Phi_\lambda(\langle X_i,u\rangle)_{i=1}^m$, however, this assumes that the coordinate system $(X_i)_{i=1}^m$ is given.  Explicitly constructing frames for saturation recovery is very difficult, but  if $\lambda>0$ is sufficiently large, then a random frame of $m$ on the order of $n$ vectors which are uniformly distributed in $\sqrt{n}S^{n-1}$ will allow for uniformly stable saturation recovery with high probability \cite{laska2011democracy,foucart-li,foucart-needham}, which solves (I) as long as $\lambda$ is greater than some fixed constant $\alpha>0$.  Thus, we are particularly interested in cases where $\lambda$ may be small and a large proportion of the coordinates will be saturated.

The fact that we assumed $X$ to be distributed on the Euclidean sphere $\sqrt{n}\S^{n-1}$ in one of our main questions is not by chance. Indeed, in applications, the declipping problem takes the following general form: given a collection of vectors $\{\mathbf{x}_i\}$, recover our original signal/vector $\mathbf{u}$ from clipped information of the form $\{\Phi_\lambda(\langle \mathbf{x}_i, \mathbf{u}\rangle)\}_i$.  Two of the most common situations involve either a single device taking multiple measurements over a period of time, such as a microphone, or multiple identical devices in different positions which each take a single measurement, such as pixels in a digital camera.  Both of these situations correspond to the measurement vectors having the same norm, and thus it is natural to consider random vectors whose distribution is contained in a sphere. 
From a mathematical point of view, the normalization assumption that $\{\mathbf{x}_i\}$ all have norm $\sqrt{n}$  allows for the likelihood of clipping to be dimension independent and makes it possible to fairly compare the problem in any dimension.

Instead of considering the problem of recovering a signal from clipped measurements, one could instead adjust the measurement devices ahead of time  to avoid the possibility of clipping.  For example, turning down the gain on a microphone would make the threshold for clipping an audio signal relatively higher, but this must be balanced by increasing the volume when  playing back the audio recording.  Although this removes the problem of clipping, it has a significant drawback as any error or noise in the recording will be scaled up as well. 
The method of \emph{unlimited sampling} takes a different approach to adjusting measurement devices to avoid clipping. In  analog-to-digital converters (ADC's), a \emph{self-reset} ADC will detect when a signal is at one end of its range and then switch to the opposite.  Therefore, instead of the graph being clipped, the graph has a jump discontinuity where it switches from its maximum to its minimum or vice-versa.  
In that direction, Bhandari, Krahmer and Raskar \cite{bhandari2020unlimited} (see also \cite{MR4815053}) proposed the usage of a non-linear mapping of the form 
\begin{equation*}
    \mathcal{M}_{\lambda}(f)  = 2\lambda \left(\frac{f}{2\lambda}+\frac{1}{2}-\left\lfloor\frac{f}{2\lambda} +\frac{1}{2}\right\rfloor -\frac{1}{2}\right).
\end{equation*}
The following images show the effect of folding and sampling on a signal.

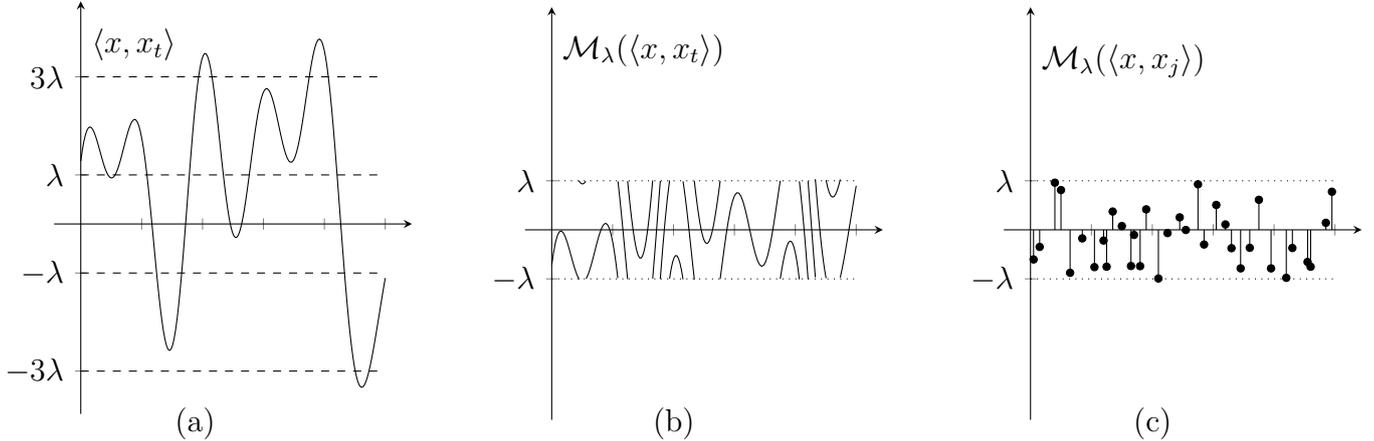
\begin{figure}[htp]\label{folding picture}
\begin{subfigure}[t]{0.3\textwidth}
\centering
\begin{tikzpicture}
\begin{axis}[width=2.2in, height=2.5in, domain=0:30, xmin=0, xmax=1, xtick={}, xticklabels={$0$,$1$}, extra x ticks={}, extra x tick labels={}, extra x tick style={xticklabel style={xshift=-0.9ex}}, ytick={-.9,-0.3,0,0.3,.9}, yticklabels={$-3\lambda$,$-\lambda$,0,$\lambda$,$3\lambda$}, ymin=-1, ymax=1.2
,legend pos=outer north east,legend cell align=left,axis lines=center,axis line style={shorten >=-10pt, shorten <=-10pt}
,x label style={at={(current axis.right of origin)},anchor=north, below=3mm,right=5mm}
]

    \addplot[domain=0:1,samples=200]{1.2*(5.5*x^2*(1-x)^2 +0.42*sin(2*3*pi*deg(x))+0.3*sin(2*2*pi*deg(x))+0.42*cos(2*5*pi*deg(x))-.1+.25*x+.42*e^(-100*(x-.9)^2)-1*e^(-150*(x-1)^2))};
\addplot[dashed,domain=0:1,samples=3]{0.3};
\addplot[dashed,domain=0:1,samples=3]{-0.3};
\addplot[dashed,domain=0:1,samples=3]{0.9};
\addplot[dashed,domain=0:1,samples=3]{-0.9};

\node[scale=1,anchor=west] at (axis cs: 0,1.1) {$\langle x, x_{t}\rangle$};

\end{axis}
\end{tikzpicture}    

(a) 
\end{subfigure}
\hfill
\begin{subfigure}[t]{0.3\textwidth}
\centering
\begin{tikzpicture}
\begin{axis}[width=2.2in, height=2.5in, domain=0:30, xmin=0, xmax=1, xtick={}, xticklabels={$0$,$1$}, extra x ticks={}, extra x tick labels={}, extra x tick style={xticklabel style={xshift=-0.9ex}}, ytick={-0.3,0,0.3}, yticklabels={$-\lambda$,0,$\lambda$}, ymin=-1, ymax=1.2
,legend pos=outer north east,legend cell align=left,axis lines=center,axis line style={shorten >=-10pt, shorten <=-10pt}
,x label style={at={(current axis.right of origin)},anchor=north, below=3mm,right=5mm}
]

\addplot[dotted,domain=0:1,samples=3]{0.3};
\addplot[dotted,domain=0:1,samples=3]{-0.3};

\addplot[domain=0:.09,samples=15]{-.6+1.2*(5.5*x^2*(1-x)^2+0.42*sin(2*3*pi*deg(x))+0.3*sin(2*2*pi*deg(x))+0.42*cos(2*5*pi*deg(x))-.1+.25*x+.42*e^(-100*(x-.9)^2)-1*e^(-150*(x-1)^2))};
\addplot[domain=.089:.112,samples=10]{1.2*(5.5*x^2*(1-x)^2+0.42*sin(2*3*pi*deg(x))+0.3*sin(2*2*pi*deg(x))+0.42*cos(2*5*pi*deg(x))-.1+.25*x+.42*e^(-100*(x-.9)^2)-1*e^(-150*(x-1)^2))};
\addplot[domain=.113:.217,samples=20]{-.6+1.2*(5.5*x^2*(1-x)^2+0.42*sin(2*3*pi*deg(x))+0.3*sin(2*2*pi*deg(x))+0.42*cos(2*5*pi*deg(x))-.1+.25*x+.42*e^(-100*(x-.9)^2)-1*e^(-150*(x-1)^2))};
\addplot[domain=.2175:.2492,samples=20]{1.2*(5.5*x^2*(1-x)^2+0.42*sin(2*3*pi*deg(x))+0.3*sin(2*2*pi*deg(x))+0.42*cos(2*5*pi*deg(x))-.1+.25*x+.42*e^(-100*(x-.9)^2)-1*e^(-150*(x-1)^2))};
\addplot[domain=.2492:.3315,samples=20]{.6+1.2*(5.5*x^2*(1-x)^2+0.42*sin(2*3*pi*deg(x))+0.3*sin(2*2*pi*deg(x))+0.42*cos(2*5*pi*deg(x))-.1+.25*x+.42*e^(-100*(x-.9)^2)-1*e^(-150*(x-1)^2))};
\addplot[domain=.3315:.3565,samples=5]{1.2*(5.5*x^2*(1-x)^2+0.42*sin(2*3*pi*deg(x))+0.3*sin(2*2*pi*deg(x))+0.42*cos(2*5*pi*deg(x))-.1+.25*x+.42*e^(-100*(x-.9)^2)-1*e^(-150*(x-1)^2))};
\addplot[domain=.3565:.387,samples=5]{-.6+1.2*(5.5*x^2*(1-x)^2+0.42*sin(2*3*pi*deg(x))+0.3*sin(2*2*pi*deg(x))+0.42*cos(2*5*pi*deg(x))-.1+.25*x+.42*e^(-100*(x-.9)^2)-1*e^(-150*(x-1)^2))};
\addplot[domain=.387:.4313,samples=10]{-1.2+1.2*(5.5*x^2*(1-x)^2+0.42*sin(2*3*pi*deg(x))+0.3*sin(2*2*pi*deg(x))+0.42*cos(2*5*pi*deg(x))-.1+.25*x+.42*e^(-100*(x-.9)^2)-1*e^(-150*(x-1)^2))};
\addplot[domain=.4313:.4685,samples=10]{-.6+1.2*(5.5*x^2*(1-x)^2+0.42*sin(2*3*pi*deg(x))+0.3*sin(2*2*pi*deg(x))+0.42*cos(2*5*pi*deg(x))-.1+.25*x+.42*e^(-100*(x-.9)^2)-1*e^(-150*(x-1)^2))};
\addplot[domain=.4683:.5533,samples=20]{1.2*(5.5*x^2*(1-x)^2+0.42*sin(2*3*pi*deg(x))+0.3*sin(2*2*pi*deg(x))+0.42*cos(2*5*pi*deg(x))-.1+.25*x+.42*e^(-100*(x-.9)^2)-1*e^(-150*(x-1)^2))};
\addplot[domain=.5533:.7505,samples=30]{-.6+1.2*(5.5*x^2*(1-x)^2+0.42*sin(2*3*pi*deg(x))+0.3*sin(2*2*pi*deg(x))+0.42*cos(2*5*pi*deg(x))-.1+.25*x+.42*e^(-100*(x-.9)^2)-1*e^(-150*(x-1)^2))};
\addplot[domain=.7505:.8133,samples=15]{-1.2+1.2*(5.5*x^2*(1-x)^2+0.42*sin(2*3*pi*deg(x))+0.3*sin(2*2*pi*deg(x))+0.42*cos(2*5*pi*deg(x))-.1+.25*x+.42*e^(-100*(x-.9)^2)-1*e^(-150*(x-1)^2))};
\addplot[domain=.8133:.8419,samples=10]{-.6+1.2*(5.5*x^2*(1-x)^2+0.42*sin(2*3*pi*deg(x))+0.3*sin(2*2*pi*deg(x))+0.42*cos(2*5*pi*deg(x))-.1+.25*x+.42*e^(-100*(x-.9)^2)-1*e^(-150*(x-1)^2))};
\addplot[domain=.8419:.8658,samples=10]{+1.2*(5.5*x^2*(1-x)^2+0.42*sin(2*3*pi*deg(x))+0.3*sin(2*2*pi*deg(x))+0.42*cos(2*5*pi*deg(x))-.1+.25*x+.42*e^(-100*(x-.9)^2)-1*e^(-150*(x-1)^2))};
\addplot[domain=.8658:.901,samples=10]{.6+1.2*(5.5*x^2*(1-x)^2+0.42*sin(2*3*pi*deg(x))+0.3*sin(2*2*pi*deg(x))+0.42*cos(2*5*pi*deg(x))-.1+.25*x+.42*e^(-100*(x-.9)^2)-1*e^(-150*(x-1)^2))};
\addplot[domain=.901:.9476,samples=10]{1.2+1.2*(5.5*x^2*(1-x)^2+0.42*sin(2*3*pi*deg(x))+0.3*sin(2*2*pi*deg(x))+0.42*cos(2*5*pi*deg(x))-.1+.25*x+.42*e^(-100*(x-.9)^2)-1*e^(-150*(x-1)^2))};
\addplot[domain=.9476:1,samples=10]{.6+1.2*(5.5*x^2*(1-x)^2+0.42*sin(2*3*pi*deg(x))+0.3*sin(2*2*pi*deg(x))+0.42*cos(2*5*pi*deg(x))-.1+.25*x+.42*e^(-100*(x-.9)^2)-1*e^(-150*(x-1)^2))};
\node[scale=1,anchor=west] at (axis cs: 0,1.1) {$\mathcal{M}_\lambda(\langle x, x_{t}\rangle)$};

\end{axis}
\end{tikzpicture}    

(b) 
\end{subfigure}
\hfill
\begin{subfigure}[t]{0.3\textwidth}
\centering
\begin{tikzpicture}
\begin{axis}[width=2.2in, height=2.5in, domain=0:30, xmin=0, xmax=1, xtick={}, xticklabels={$0$,$1$}, extra x ticks={}, extra x tick labels={}, extra x tick style={xticklabel style={xshift=-0.9ex}}, ytick={-0.3,0,0.3}, yticklabels={$-\lambda$,0,$\lambda$}, ymin=-1, ymax=1.2
,legend pos=outer north east,legend cell align=left,axis lines=center,axis line style={shorten >=-10pt, shorten <=-10pt}
,x label style={at={(current axis.right of origin)},anchor=north, below=3mm,right=5mm},ylabel={$\mathcal{M}_\lambda(\langle x, x_{j}\rangle)$}
]

\addplot[dotted,domain=0:1,samples=3]{0.3};
\addplot[dotted,domain=0:1,samples=3]{-0.3};

\addplot+[ycomb,color=black,mark size=1.375pt,mark options={fill=black}] coordinates 
{ 
(	0.01	,	-0.18110927496074	)
(	0.03	,	-0.104556628340579	)
(	0.08	,	0.288434401953001	)
(	0.1	,	0.243220691732511	)
(	0.13	,	-0.262819085739908	)
(	0.17	,	-0.0523030485888684	)
(	0.21	,	-0.228046438442514	)
(	0.24	,	-0.06559312159928	)
(	0.25	,	-0.22546875	)
(	0.33	,	-0.221397815815752	)
(	0.38	,	0.125343257075586	)
(	0.42	,	-0.297689678031843	)
(	0.49	,	0.0763265322388639	)
(	0.51	,	-0.000873748757645083	)
(	0.55	,	0.278341697799505	)
(	0.61	,	0.152210916860581	)
(	0.64	,	0.0329790983779996	)
(	0.66	,	-0.112310166590181	)
(	0.72	,	-0.110050122004388	)
(	0.79	,	-0.235862253305047	)
(	0.84	,	-0.293552304594592	)
(	0.92	,	-0.225866126759093	)
(	0.97	,	0.0430613061356769	)
(	0.99	,	0.233018928859242	)
(	0.27	,	0.11191660981464	)
(	0.36	,	-0.222207038458318	)
(	0.45	,	-0.0197571873760678	)
(	0.75	,	0.183714106080737	)
(	0.86	,	-0.111100537351346	)
(	0.91	,	-0.196698885255246	)
(	0.3	,	0.0228546183494194	)
(	0.34	,	-0.031249919280266	)
(	0.69	,	-0.235676490801371	)
(	0.57	,	-0.0897848525075712	)

};

\end{axis}
\end{tikzpicture}

(c)

\end{subfigure}

\caption{Effect of folding and sampling: (a) Unfolded signal (b) Folded signal (c) Folded and randomly sampled signal.} \label{fig:saturation-continuous-fold}
\end{figure}

 We have used a smaller value for $\lambda$ in Figure 2 than in Figure 1 in order to illustrate what happens when parts of the signal are folded multiple times.  One can see how to obtain the original signal $(\langle x,x_t\rangle)_{t\in[0,1]}$ from the folded signal $(\mathcal{M}_\lambda(\langle x,x_t\rangle))_{t\in [0,1]}$ by sequentially unfolding each connected piece.  This unfolding procedure requires knowing each point where the signal crosses the threshold, but unfortunately, this information is lost when randomly sampling.  When looking at the graph of a clipped and randomly sampled signal as in Figure 1, one can tell where the signal was clipped and one knows the exact value for the unclipped portions.  On the other hand, in Figure 2 the graph of the folded and randomly sampled signal gives no information on which values were folded.    There are algorithms for stably recovering a signal from random saturated measurements \cite{alharbi2024declipping,foucart-li,foucart-needham,laska2011democracy}, however we are unaware of any algorithm for stably recovering a signal from random folded measurements.  In contrast to this, there exist algorithms outside of the random setting for recovering signals from folded measurements.  In particular, band-limited functions may be stably recovered from folded measurements when sampled above the Nyquist-rate \cite{bhandari2020unlimited,RO}.

Notice that both the declipping and unlimited sampling problems involve a parameter $\lambda$. This parameter ties the two problems together, as we wish to understand the precise dependence on $\lambda$ and answer the same questions that we asked in the case of the clipping operator only now  applying the non-linear mapping $\mathcal{M}_\lambda$ on the samples instead. More precisely, we wish to know how many samples are necessary in order to stably recover the signal, to identify the optimal stability constant, and to find general conditions on the distribution of the samples for such properties to hold. 
Since these three questions are quite different in nature, we shall divide the discussion of our main results into several parts, where at each step we highlight the similarities and differences between unlimited sampling and declipping.
\subsection{Declipping and unlimited sampling for general distributions} The first part of the manuscript deals with the question of declipping and unlimited sampling for \emph{general} distributions. That is, what are the mildest (distributional) conditions that we may impose on a random vector $X$ so that it satisfies stable recovery after clipping or folding with $m \ge c(\lambda)n$ measurements? 
Our first main theorem deals with the case of declipping and characterizes this in an \emph{almost sharp} way.
\begin{theorem}\label{General distributions} Let $\mathcal{V}$ be a set of random vectors in Euclidean space such that the following two conditions hold.  First, for all $\varepsilon>0$ there is a $\delta > 0$ so that for all $X\in \mathcal{V}$ if $X$ is a random vector in $\R^n$ then
\begin{equation}\label{eq:small-ball-og}
\sup_{v \in \S ^{n-1}} \mathbb{P}(|\langle X,v\rangle| \le \delta ) \le \varepsilon,
\end{equation}
Second,  for all $\delta > 0$ there is an $\tilde{\varepsilon}>0$ such that for all $X\in \mathcal{V}$ if $X$ is a random vector in $\R^n$ then
\begin{equation}\label{eq:reverse-small-ball}
\inf_{u \in \S ^{n-1}} \mathbb{P}(|\langle X,u \rangle| \le \delta) \ge \tilde{\varepsilon}.
\end{equation}
Then there are $C(\lambda),c(\lambda)>0$  such that, with probability at least $1 - e^{- c(\lambda)m},$ we have that, whenever $X\in \mathcal{V}$ is a random vector in $\R^n$ and $m \ge C(\lambda) \cdot n$ and $u,v \in \B_{\R^n},$
\[\frac{1}{m}\sum_{i=1}^m \left|\Phi_{\lambda}(\langle X_i,u\rangle) - \Phi_{\lambda}(\langle X_i,v\rangle)\right|^2 \ge c(\lambda)\|u-v\|_{2}^2,\]
where $X_i$ are i.i.d.~copies of the vector $X$.
\end{theorem}

In order to prove Theorem \ref{General distributions}, we employ VC-dimension-based arguments. Our proof strategy is inspired by the small ball method \cite{koltchinskii2015bounding,mendelson2015learning}, but, surprisingly, it requires the new `reverse-type' small ball condition \eqref{eq:reverse-small-ball}, which, as we shall see below, is in a sense also \emph{necessary}. 

The first step in the proof of Theorem \ref{General distributions} is to estimate the empirical average of clipped measurements from below by an associated Bernoulli selector process -- that is, one that has a sum of indicators of a certain event. This follows by bounding the empirical average from below whenever we can guarantee that the clipping operator does not clip any of the terms. With that in mind, we are left with the task of obtaining a uniform quantified version of the law of large numbers. This can be potentially achieved in different ways, but one which we highlight as especially interesting is through the usage of VC dimension bounds. 

In succinct form, the VC dimension of a set measures the complexity of it in a combinatorial sense. In addition, the VC dimension is convenient to handle unions and intersections of sets, as we shall precisely clarify below. Since the sets of interest are intersections of up to three hyperplanes, accurate bounds on the VC dimensions of such sets are relatively simple to obtain; we cover the most relevant preliminary results in that regard in Section~\ref{sec:prelim}.


We note, furthermore, that the condition \eqref{eq:reverse-small-ball} is, as a matter of fact, \emph{necessary} in order to obtain a result such as Theorem \ref{General distributions} in the continuous case -- see Remark \ref{rmk:rsb-needed}. This marks a first fundamental difference between the problem of declipping and the by now well-studied problem of phase retrieval: reverse small ball conditions seem to play no role in the latter, while they are seen to be crucial for the former. 

In Theorem~\ref{General distributions} we have identified situations where we are able to perform $\lambda$-declipping for all $\lambda>0$. On the other hand, a natural question is to identify which distributions satisfy $\lambda$-stable declipping for \emph{some} $\lambda > 0$. As a matter of fact, we are able to characterize this result \emph{exactly} in the continuous setting: 

\begin{theorem}\label{thm:general-some-lambda}
Let $\mu$ be a probability measure and $1< p<\infty$. The following properties are equivalent for a subspace $E$ of $L^p(\mu)$.
\begin{enumerate}
\item[\normalfont(I)] $B_E$ does stable $\lambda$-declipping for some $\lambda>0$, i.e., there exists $\lambda, C>0$ so that for all $f,g\in B_E$ we have
\begin{equation*}
\|f-g\|_{L^p}\leq C \| \Phi_\lambda(f) -\Phi_\lambda(g)\|_{L^p}.
\end{equation*}
\item[\normalfont(II)] The $L^1$ and $L^p$ norms are equivalent on $E$.
\end{enumerate}
\end{theorem}

By considering the space of variables of the form $\langle X,u\rangle$, where $X$ is a fixed random vector and $u \in \R^n$, we have the following immediate consequence of the previous result: 

\begin{corollary}\label{declipping corollary some} Let $X$ be a random vector in $\R^n$. Then $X$ does $\lambda$-stable declipping for some $\lambda > 0$, in the sense that there exist $\lambda, C >0$ so that for all $u,v \in \S^{n-1}$, we have 
\[
\|u-v\|_2^2 \le C \cdot \mathbb{E} \left| \Phi_{\lambda}(\langle X,u\rangle) - \Phi_{\lambda}(\langle X,v\rangle) \right|^2,
\]
if, and only if, the \emph{weak small ball property} 
\[
\sup_{v \in \S^{n-1}} \mathbb{P}\left( |\langle X,v\rangle| \le \delta \right) \le \varepsilon
\]
holds for some $\varepsilon, \delta > 0$. 
\end{corollary}

The above result motivates the use of the \emph{strong small ball property} \eqref{eq:small-ball-og} in Theorem~\ref{General distributions}, as it is a natural analogue of the weak small ball property which allows one to upgrade the conclusion of Corollary~\ref{declipping corollary some}  from \emph{some}  $\lambda>0$ to \emph{all} $\lambda>0$. On the other hand, characterizing the exact conditions on $X$ which ensure that stable recovery holds for a \emph{fixed} $\lambda>0$ seems to be an interesting open problem.

As a consequence of the proof of Theorem \ref{General distributions}, we are able to prove the following result in the \emph{sparse setting}. Here, we use $T_s$ to denote the set of $s$-sparse vectors in the ball of $\R^n$.
\begin{theorem}\label{thm:sparse-gen-dist} Let $\mathcal{V}$ be a class of random vectors in Euclidean space satisfying \eqref{eq:small-ball-og} and \eqref{eq:reverse-small-ball}. Then there are $C(\lambda),c(\lambda)>0$ such that, with probability at least $1 - e^{-c(\lambda) n},$ we have that, whenever $m \ge C(\lambda) \cdot s \cdot \log \left( \frac{e \cdot n}{s}\right)$ and $u,v \in T_s,$
\[\frac{1}{m}\sum_{i=1}^m \left|\Phi_{\lambda}(\langle X_i,u\rangle) - \Phi_{\lambda}(\langle X_i,v\rangle)\right|^2 \ge c(\lambda)\|u-v\|_{2}^2,\]
where $X_i$ are i.i.d.~copies of the vector $X$.
\end{theorem}

This problem has been brought to light in \cite{foucart-li, foucart-needham,  laska2011democracy}, where the authors were interested in the problem of recovering a vector $x$ which is $s$-sparse from measurements of the form $\Phi_{\lambda}(\langle A_i ,x\rangle)$. The importance of this new problem comes from the natural goal of blending compressed sensing ideas \cite{candes2006stable,candes2006robust,donoho2006compressed} together with recovery from clipped measurements. 

Finally, we are able to present the following result for the question of stability for the unlimited sampling problem. For simplicity, we state it in the non-sparse setting.

\begin{theorem}\label{thm:mod-sharp-try} Let $X$ be a random vector in $\R^n$ for which there are $a,c>0$ such that 
\begin{equation}\label{eq:small-strip}
\inf_{u \in \S^{n-1}} \mathbb{P}\left( |\langle X,u\rangle|\in [a,1)\right) > c, 
\end{equation}
and such that, if 
\[
\Omega_a(\delta) \coloneqq \cup_{m \in \Z} \left[ (a + 2m)\delta, (1-a + 2m)\delta\right],
\]
then for all $\delta \in (0,1)$
\begin{equation}\label{eq:mod-cond}
\inf_{u \in \S^{n-1}} \mathbb{P} \left( |\langle X,u\rangle| \in \Omega_a(\delta)\right) \ge c(\delta) > 0.
\end{equation}
Then, for each $\lambda > 0$ there are $C(\lambda),c(\lambda)>0$  such that, with probability at least $1-e^{-c(\lambda)m}$, we have that whenever $m \ge C(\lambda) \cdot n,$ 
\[
\frac{1}{m} \sum_{i=1}^m |\mathcal{M}_\lambda(\langle X_i,u\rangle) - \mathcal{M}_\lambda(\langle X_i,v\rangle)|^2 \ge c(\lambda)\|u-v\|_2^2,
\]
for all $u,v \in \B_{\R^n}$. 
\end{theorem}
In a nutshell, the assumptions \eqref{eq:small-strip} and \eqref{eq:mod-cond} state that the marginals of the vector $X$ along any direction $u$ do not jump ``too much'', i.e, they have some regularity. In particular, for a centered isotropic $X$ ($\E \langle X,u\rangle^2=1$), it suffices to require that $\langle X,u\rangle$ has a bounded continuous density. We will discuss why the above conditions arise naturally in Section~\ref{sec:general-results}.

\subsection{Sharp stability for \texorpdfstring{$X \sim_d \text{Unif}(\sqrt{n}{\S^{n-1}})$}{}} We now deal with our second main question; that is, the matter of more \emph{particular} distributions and obtaining sharp estimates for the associated $\lambda$-dependent parameters. More explicitly, we are interested in the particular case when $X$ is uniformly distributed on the sphere $\sqrt{n} \S^{n-1}$ and wish to identify the \emph{sharpest} stability constants as well as the \emph{minimal} number of samples needed in order to achieve such estimates. As it will turn out, the dependencies that we will prove are not only sharp for the uniform distribution but are also best possible for \emph{any} frame of normalized vectors. Moreover, in the course of the proof, we will identify general conditions on the distribution which ensure that the sharp dependence can be achieved.

When considering the problem of declipping, there is a big distinction between large and small values of $\lambda$.  If $\lambda$ is sufficiently larger than $1$, then it is possible to simply throw out all the clipped values and stably recover a vector from the unclipped coordinates.  However, if $\lambda$ is small then declipping is significantly more difficult and throwing out the clipped values will substantially reduce the stability of recovery.  The following theorem was proven in \cite{laska2011democracy} for the case of sufficiently large $\lambda$.

\begin{theorem}\label{thm:main-laska}[\cite{laska2011democracy}]
There exist absolute constants $C,c>0$ and $\alpha>0$  for which the following holds: Let $\lambda>\alpha$ and $m\ge Cn $. If $X_1,\ldots, X_m$ are independent copies of a random vector $X$ uniformly distributed on the Euclidean sphere $\sqrt{n}\S^{n-1}$  then, with probability at least $1-e^{-c m}$,
\begin{equation*}\label{Stability bound 1}
   \left( \frac{1}{m}\sum_{i=1}^m \left|\Phi_{\lambda}(\langle X_i,u\rangle) - \Phi_{\lambda}(\langle X_i,v\rangle)\right|^2\right)^{1/2} \ge c\|u-v\|_{2},
\end{equation*}
for all $u,v$ in the unit ball $\B_{\R^n}$. 
\end{theorem}

Our first main result is to prove the analogue of Theorem \ref{thm:main-laska} for the case of small $\lambda$.  In this setting, both the number of samples and the stability of recovery necessarily depend on $\lambda$.  We prove the following theorem which gives that stable recovery is possible when taking $m=O(\lambda^{-1} \log(1/\lambda)n)$ samples. 


\begin{theorem}\label{thm:main-vc}
There exist absolute constants $C,c>0$  for which the following holds: Let $0<\lambda \le  1/2$ and $m\ge C\lambda^{-1} \log(1/\lambda) n $. If $X_1,\ldots, X_m$ are independent copies of a random vector $X$ uniformly distributed on the Euclidean sphere $\sqrt{n}\S^{n-1}$ then, with probability at least $1-e^{-c\lambda \frac{m}{\log(1/\lambda)}}$,
\begin{equation*}\label{Stability bound lambda3/2}
    \left(\frac{1}{m}\sum_{i=1}^m \left|\Phi_{\lambda}(\langle X_i,u\rangle) - \Phi_{\lambda}(\langle X_i,v\rangle)\right|^2\right)^{1/2} \ge c \lambda^{3/2}\|u-v\|_{2},
\end{equation*}
for all $u,v$ in the unit ball $\B_{\R^n}$. 
\end{theorem}
{ 
In Theorem \ref{thm:main-vc} we have that the stability of $\lambda$-saturation recovery on the unit ball scales according to $\lambda^{3/2}$. The following theorem gives that not only is this dependence on $\lambda$ optimal for $(X_i)_{i=1}^m$ being independent and uniformly distributed on $\sqrt{n}\S^{n-1}$, but that the factor of $\lambda^{3/2}$ is the best possible dependence on $\lambda$ for any collection of vectors in $\sqrt{n}\S^{n-1}$.

\begin{theorem}\label{thm:best-lam}
    There exists an absolute constant $\beta>0$ for which the following holds.  Let $0<\lambda \le  1$ and $m\ge n$.  If  $x_1,\ldots, x_m$ are any sequence of vectors in $\sqrt{n}\S^{n-1}$ then there exists $u,v\in  \B_{\R^n}$ so that 
\begin{equation*}\label{Stability bound lambda3/2 best}
    \left(\frac{1}{m}\sum_{i=1}^m \left|\Phi_{\lambda}(\langle x_i,u\rangle) - \Phi_{\lambda}(\langle x_i,v\rangle)\right|^2\right)^{1/2} < \beta \lambda^{3/2}\|u-v\|_{2}.
\end{equation*}    
\end{theorem}
}

As stated previously, the fact that we have chosen the vectors $X_i$ to be uniformly distributed on the Euclidean sphere is motivated by  applications. However, as will become evident below, our proof of Theorem~\ref{thm:main-vc} is not restricted to this specific distribution. In fact, we will identify rather precise necessary and sufficient conditions on the distribution in order for Theorem~\ref{thm:main-vc}  to hold.  

The proof of Theorem \ref{thm:main-vc} relies on an $\varepsilon$-net type argument. Indeed, by leveraging the precise geometric constraints given by the distributions of the vector $X$ and by $u,v \in \B_{\R^n},$ we are able to construct well-controlled $\varepsilon$-nets -- a feature which reduces the number of samples needed leading to the optimal dependence with respect to $\lambda$. 

Although it might be surprising at first that an $\varepsilon$-net estimate is more accurate than the more technically involved argument via VC-dimension, we remark on the fact that this is mainly due to the precise geometric nature of the data distribution while VC-dimension is a complexity that relies on a combinatorial nature of the function class conditioned on the data distribution.

Theorems \ref{thm:main-laska} and \ref{thm:main-vc} give  stability bounds for recovering a vector in the unit ball from clipped random measurements.  We now state the analogous theorem for folded measurements.

\begin{theorem}\label{T:folding}  There exist absolute constants $C,c> 0$  for which the following holds: Let $X_1,\dots,X_m$ be independent copies of a random vector $X$ uniformly distributed on the Euclidean sphere $\sqrt{n} \S^{n-1}$.  If  $\lambda\geq 1/2$ and  $m \ge C n$ then  with probability at least $1 - e^{-c m}$, we have for all $u,v \in \B_{\R^n}$ that
\[
\left(\frac{1}{m} \sum_{i=1}^m |\mathcal{M}_\lambda(\langle X_i,u\rangle) - \mathcal{M}_\lambda(\langle X_i,v\rangle)|^2\right)^{1/2} \ge c\|u-v\|_2.
\]

If  $0<\lambda< 1/2$ and  $m \ge C\log(1/\lambda) n$ then  with probability at least $1 - e^{-c m/\log(1/\lambda)}$, we have for all $u,v \in \B_{\R^n}$ that
\[
\left(\frac{1}{m} \sum_{i=1}^m |\mathcal{M}_\lambda(\langle X_i,u\rangle) - \mathcal{M}_\lambda(\langle X_i,v\rangle)|^2\right)^{1/2} \ge c\lambda\|u-v\|_2.
\]
\end{theorem}
Notice that in Theorem~\ref{T:folding} the stability constant now scales like $\lambda$ (rather than $\lambda^{3/2}$) and the number of samples is significantly lower than in the analogous declipping problem. As before, having a stability constant proportional to $\lambda$ is optimal for this problem not only for the uniform distribution but for all collections of vectors in $\sqrt{n} \S^{n-1}$ -- see Remark \ref{rmk:sharpness}. Furthermore, it will be evident from the proof of Theorem~\ref{T:folding} that the sharp dependence on $\lambda$ is not restricted to the uniform case, but can in fact be achieved for a rather wide  class of distributions.  

Similarly to the previous subsection, the results in this subsection have sparse counterparts. For simplicity, we state our sparse recovery result for clipped measurements:

\begin{theorem}\label{thm:main-sparse}
Let $\lambda > 0$ be an absolute constant and let $T_s$ be the subset of the unit Euclidean ball consisting of all $s$-sparse vectors. Let $X_1,\ldots, X_m$ be independent copies of a random sub-gaussian vector $X$. There exists absolute constants $C,c_1 > 0$ for which the following hold. 

\begin{enumerate}
    \item[\emph{(I)}] If $m\ge C \lambda^{-1}\log\left( \frac{1}{\lambda}\right) s \log(en/s)$ then,  with probability at least $1-e^{-c \lambda m}$, we have 
    \begin{equation*}
    \inf_{(u,v) \in T_s\times T_s}\left(\frac{1}{m\|u-v\|_2^2}\sum_{i=1}^m \left|\Phi_{\lambda}(\langle X_i,u\rangle) - \Phi_{\lambda}(\langle X_i,v\rangle)\right|^2\right)^{1/2} \ge c_1 \lambda^{3/2}.
\end{equation*}

    \item[\emph{(II)}] If $m \ge C \lambda^{-3} \log\left( \frac{1}{\lambda}\right) s \log(en/s)$ then, with probability at least $1 - e^{-c \lambda m}$, we have 
    \begin{equation*}
    \inf_{(u,v) \in \Sigma_s\times \Sigma_s}\left(\frac{1}{m\|u-v\|_2^2}\sum_{i=1}^m \left|\Phi_{\lambda}(\langle X_i,u\rangle) - \Phi_{\lambda}(\langle X_i,v\rangle)\right|^2\right)^{1/2} \ge c_1 \lambda^{3/2},
\end{equation*}
    where $\Sigma_s = \{ x \in \B_{\R^n} \colon \|x\|_1 \le \sqrt{s}\|x\|_2\}$ denotes the set of effectively sparse vectors in the unit ball. 
\end{enumerate}

\end{theorem}

{As an example of a sparse counterpart result of Theorem \ref{T:folding}, one can show that for $m\gtrsim \log(1/\lambda)n$ (or $m\gtrsim \lambda^{-2}\log(1/\lambda)n$) the same results of Theorem \ref{T:folding} (up to an absolute constant) hold for all $u,v \in T_s$ (or $u,v \in \Sigma_s$). For the sake of simplicity, we omit the proof as it is similar to that of Theorem \ref{thm:main-sparse}.

Although, like in Theorem \ref{thm:main-vc}, Part I of Theorem \ref{thm:main-sparse} admits a proof that resorts to VC-dimension considerations, we highlight that we need different ideas when considering the set of \emph{effectively sparse} vectors which is the convex hull of the set of $s$-sparse vectors. The main reason why we need to forego of VC-dimension techniques is due to the fact that VC-dimension fails to capture `effective-sparsity': we only know, a priori, a \emph{Gaussian width bound} on that set, and hence the VC-dimension estimates are now obsolete, since the exact algebraic structure of the sets under scrutiny is no longer the desired object. 

We remark that the increase in the number of samples in the effectively sparse case is not artificial and is in fact systematically related to the complexity of the set of effectively sparse vectors. We will discuss this in detail below; see Remark~\ref{remark on samples}.}
\subsection{Sparse results connected to one-bit compressed sensing} In the final set of main results of this manuscript, we highlight applications to other adjacent problems, where our methods also yield \emph{sharp bounds}. 

Our next result deals with the problem of recovery from clipped measurements of sparse vectors on the \emph{sphere} rather than the recovery of clipped measurements of vectors in the \emph{ball} which we studied in previous subsections. Somewhat surprisingly, this reduction happens to have a profound impact on the stability guarantees.  One main reason to consider this problem is that, by taking $\lambda \to \infty$, one obtains the usual problem of sparse recovery in compressed sensing -- as explored in the references above -- and when one takes a normalized limit when $\lambda \to 0$ one formally obtains the analogously famous problem of \emph{one-bit sparse recovery} \cite{boufounos20081,plan2013one}, which asks for recovery conditions on sparse vectors $x$ from measurements of the form $\text{sign}(\langle A_i,x\rangle)$.

In this direction, we highlight the following result. 

\begin{theorem}\label{thm:spherical} Let $\lambda > 0$ be an absolute constant, and let $X$ be a random vector uniformly distributed on the sphere $\sqrt{n} \S^{n-1}.$ Then there are absolute constants $C,c>0$ for which the following holds: If $X_1,\dots,X_m$ denote independent copies of $X$, and $m \ge C \lambda^{-1} \log\left( \frac{1}{\lambda}\right)n,$ then, with probability at least $1 - e^{- c \lambda m}$, we have 
\[
\inf_{u,v \in \S^{n-1}} \left(\frac{1}{m\|u-v\|_2} \sum_{i=1}^m |\Phi_{\lambda}(\langle X_i,u\rangle) - \Phi_{\lambda}(\langle X_i,v\rangle)|^2 \right)^{1/2}\ge c \lambda.
\]
\end{theorem}
The key point in Theorem~\ref{thm:spherical} is that, if one restricts to $u,v$ in the sphere rather than in the unit ball, the stability constant is improved by a factor of $\lambda^{1/2}$. As a direct consequence of the method of proof of Theorem \ref{thm:spherical}, we are able to obtain the following result about \emph{sparse} and \emph{effectively sparse} recovery: 

\begin{corollary}\label{cor:sparse-rec} Let $\lambda > 0$ be an absolute constant, and let $X$ be a random vector uniformly distributed on the sphere $\sqrt{n} \S^{n-1}.$ Let also $T_s$ denote the subset of the Euclidean \emph{unit sphere} consisting of $s$-sparse vectors and $\Sigma_s$ the set of effectively $s$-sparse vectors. Then there are absolute constants $C,c>0$ for which the following assertions hold for $m$ independent copies $X_1,\ldots,X_m$ of $X$. 

\begin{enumerate} 
\item[\emph{(I)}] If $m \ge C \lambda^{-1} \log\left( \frac{1}{\lambda}\right) s \log\left( \frac{en}{s}\right)$, then with probability at least $1-e^{-c\lambda m} $ we have  
\[
\inf_{u,v \in T_s} \left(\frac{1}{m\|u-v\|_2} \sum_{i=1}^m |\Phi_{\lambda}(\langle X_i,u\rangle) - \Phi_{\lambda}(\langle X_i,v\rangle)|^2\right)^{1/2} \ge c \lambda;
\]

\item[\emph{(II)}] If $m \ge C \lambda^{-3} \log\left( \frac{1}{\lambda}\right) s \log\left( \frac{en}{s}\right)$, then with probability at least $1-e^{-c\lambda^2 m} $ we have  
\[
\inf_{u,v \in \Sigma_s \cap \S^{n-1}} \left(\frac{1}{m\|u-v\|_2} \sum_{i=1}^m |\Phi_{\lambda}(\langle X_i,u\rangle) - \Phi_{\lambda}(\langle X_i,v\rangle)|^2\right)^{1/2} \ge c \lambda.
\]
\end{enumerate}
\end{corollary}

In the special case of effectively sparse vectors that are ``far apart", these results may be obtained as a generalization of the local binary embedding in one-bit compressed sensing \cite[Theorem 2.4]{oymak2015near}, which itself improves a line of research pioneered by Plan and Vershynin \cite{plan2013one}\footnote{The result is stated for Gaussians but it can be adapted to the case of the uniform distribution on the sphere.}. As a matter of fact, as a consequence of the method of proof of Corollary \ref{cor:sparse-rec}, we have the following stability result for one-bit recovery: 

\begin{corollary}\label{cor:one-bit-cs} There are absolute constants $C,c>0$ such that the following assertion holds. Let $\delta > 0$ be an absolute constant, and let $X$ denote a random vector uniformly distributed on the sphere $\sqrt{n}\S^{n-1}.$ Let $\Sigma_s$ denote the set of effectively sparse vectors in $\R^n$.  If $X_1,X_2,\dots,X_m$ denote independent copies of $X$ and $m \ge C\cdot \delta^{-3} \log\left( \frac{1}{\delta}\right) s \, \log \left( \frac{en}{s}\right)$, then with probability at least $1 - e^{-c \delta^2 m}$, we have 
\[
\inf_{0<\lambda < \delta} \inf_{\substack{{u,v \in \Sigma_s \cap \S^{n-1}}\\{\|u-v\|_2 \ge \delta}}} \frac{1}{m\|u-v\|_2^2} \sum_{i=1}^m \frac{|\Phi_\lambda(\langle X_i,u\rangle) - \Phi_\lambda(\langle X_i,v\rangle)|^2}{\lambda^2} \ge c. 
\]
In particular, it follows that, for any $m \ge C\cdot \delta^{-3} \log\left( \frac{1}{\delta}\right) s \, \log \left( \frac{en}{s}\right)$, with probability at least $1 - e^{-c \delta^2 m}$, we have
\[
\inf_{\substack{{u,v \in \Sigma_s \cap \S^{n-1}}\\{\|u-v\|_2 \ge \delta}}} \frac{1}{m \|u-v\|_2} \sum_{i=1}^m |\text{\emph{sign}}(\langle X_i,u\rangle) - \text{\emph{sign}}(\langle X_i,v\rangle)|^2 \ge c.
\]
\end{corollary}

This manuscript is organized as follows. In Section \ref{sec:prelim}, we collect the main facts that we shall need for the remainder of the manuscript from both the theory of empirical processes and VC-dimension. Then, we prove the stable recovery-type results for general distributions (Theorems \ref{General distributions}, \ref{thm:mod-sharp-try} and \ref{thm:main-sparse}) in Section \ref{sec:general-results}. The results about sharp recovery in the case of $X$ distributed as the uniform measure on $\sqrt{n}\S^{n-1}$ (and generalizations thereof) are then proved in Section \ref{sec:declipping}, and finally, the results pertaining to sparse recovery and its applications to one-bit compressive sensing are proved in Section \ref{sec:one-bit}. 

\section{Preliminaries}\label{sec:prelim}

\subsection{Empirical processes and concentration inequalities} We start by stating some general, basic facts about concentration inequalities. For brevity, we do not provide a proof of any of these facts here, and instead refer the reader to the respective references. 

The first fact that we highlight is the so-called \emph{bounded difference inequality}, which stresses that, whenever a function has bounded differences in each coordinate, then it satisfies a good concentration inequality. To this end, we say that a function $f: \R^n \to \R$ satisfies the bounded difference condition if for every index $i$ and any $x_1,\ldots,x_n$ and $x_i'$, we have that 
\[
|f(x_1,\ldots,x_i,\ldots,x_n)-f(x_1,\ldots,x_i',\ldots,x_n)|\le c_i,
\]
for some constant $c_i>0$.
\begin{proposition}[Bounded difference inequality; Theorem~2.9.1~in~\cite{vershynin2018high}]\label{prop:bounded-differences} Let $X_1, \ldots, X_n$ be independent random variables. Let $f: \R^n \to \R$ be a measurable function satisfying the bounded difference condition with constants $c_1,\ldots,c_n$. Then, for any $t>0$, we have
$$
\mathbb{P}\left( f(X_1,\ldots,X_n)-\mathbb{E} f(X_1,\ldots,X_n) \geq t\right) \leq \exp \left(-\frac{2 t^2}{\sum_{i=1}^n c_i^2}\right).
$$
\end{proposition}

Throughout this manuscript, especially in the proof of Theorem \ref{thm:main-sparse}, it will be useful to know what it means for a random variable to be \emph{sub-gaussian}:

\begin{definition}[Sub-gaussian norm $\psi_2$] The \emph{sub-gaussian norm} of a random variable $Z$ is defined by
\[
\|Z\|_{\psi_2} \coloneqq \inf \left\{ t>0 \colon \,\,\mathbb{E} \exp\left(Z^2/t^2\right) \le 2 \right\}.
\]
A random variable $Z$ is called \emph{sub-gaussian} if $\|Z\|_{\psi_2}$ is finite. A random vector $X \in \R^n$ is sub-gaussian if there exists a $K>0$ for which
\begin{equation*}
    \sup_{u\in \S^{n-1}}\|\langle X,u\rangle\|_{\psi_2}\le K.
\end{equation*}
The smallest $K>0$ for which such estimate holds is called the sub-gaussian norm of the random vector $X$.
\end{definition}

Regarding sub-gaussian random vectors, we will always assume that the sub-gaussian norm is an absolute constant. We need, in that direction, the following standard result about sub-gaussian random vectors.

\begin{proposition}[Theorem~4.6.1 in \cite{vershynin2018high}]\label{prop:vershynin-subgauss} Let $X_1,\dots,X_m$ be independent, isotropic, mean zero sub-gaussian random vectors. Then, with probability at least $1 - e^{-c m}$, we have 
\begin{equation}\label{eq:upper-frame-bound} 
\frac{1}{m} \sum_{i=1}^m |\langle X_i, w\rangle|^2 \le 2 \|w\|_2^2,  
\end{equation} 
for all $w \in \R^n.$
\end{proposition}

\vspace{2mm}

\subsection{VC dimension bounds and a quantitative laws of large numbers} In this subsection, we will explore the main facts pertaining to the theory of VC-dimension of Boolean functions\footnote{There is a natural correspondence between VC dimension of sets and Boolean functions via $\{x:f(x)=1\}$.}. In heuristic terms, one may think of the VC-dimension of a class of Boolean functions as a measure of complexity of the class under consideration, in the sense that, informally, it plays the role of covering numbers for the class of Boolean functions. 

In order to define the VC-dimension, we need to first define what it means for a class of Boolean functions to shatter a given set. Note that, intuitively, this means that the class of functions separates the points of $S$ into binary ``buckets'' as well as one wishes. 

\begin{definition} Let $S \subset \R^n$ be a finite set, and let $\mathcal{F}$ denote a class of Boolean functions defined on $\R^n$. We say that the class $\mathcal{F}$ \emph{shatters} the set $S$ if, given any function $g:S \to \{0,1\}$, there is $f \in \mathcal{F}$ such that $f|_S = g$. 
\end{definition}

With the definition of shattering in hand, the definition of VC-dimension of a class is simply the cardinality of the largest set that the class shatters.

\begin{definition}[VC-dimension] Let $\mathcal{F}$ denote a class of Boolean functions defined on $\R^n$. We define the \emph{Vapnik-Chervonenkis (VC) dimension} $vc(\mathcal{F})$ of the class $\mathcal{F}$ as the largest positive integer $n$ for which there is a set $S \subset \R^n$ with cardinality $n$ that is shattered by $\mathcal{F}$.  
\end{definition}

In this present work, the concept of VC-dimension will mainly enter our considerations through one specific avenue, which is in a quantified version of the law of large numbers for classes of Boolean functions. The main tool which allows us to draw such a connection is the content of the next proposition. 

\begin{proposition}[Theorem~8.3.23~in~\cite{vershynin2018high}]\label{prop:vc-quantif-lln} Let $\mathcal{F}$ be a class of Boolean functions and let $X_1,\ldots,X_m$ be i.i.d.~copies of a random vector $X \in \R^n$. Then there exists a constant $C>0$ for which
\begin{equation*}
\mathbb{E}\sup_{f\in \mathcal{F}}\left|\frac{1}{m}\sum_{i=1}^m f(X_i)-\mathbb{E}f(X)\right|\le C\sqrt{\frac{vc(\mathcal{F})}{m}}.
\end{equation*}
\end{proposition}

It is also useful to state the following result on the VC-dimension of classes of intersections and unions. This is the content of the next result. 

\begin{proposition}[Theorem~1.1~in~\cite{Wellner-vdV}]\label{prop:intersection} Let $\mathcal{F}_1,\mathcal{F}_2,\dots,\F_k$ denote classes of Boolean functions, each of which with respective VC-dimensions $v_1,v_2,\dots,v_k.$ 
     Let 
    $$\F = \{ \mathds{1}(S_1 \cap S_2 \cap \cdots \cap S_k) \, \colon \mathds{1}(S_i) \in \F_i, \, i=1,2,\dots,k\}$$
    and 
    $$\mathcal{G} = \{ \mathds{1}(S_1 \cup S_2 \cup \cdots \cup S_k) \, \colon \mathds{1}(S_i) \in \F_i, \, i=1,2,\dots,k\}.$$
    Then we have 
    \[
    \max\{ vc(\F),vc(\mathcal{G})\} \le c_1 \cdot \left( \sum_{i=1}^k v_i \right) \log(c_2 \cdot k),
    \]
    where $c_1,c_2>0$ are two absolute constants.
\end{proposition}

We refer the reader to \cite{Wellner-vdV} for a proof of this fact. We also highlight another consequence of the techniques of proof of Proposition \ref{prop:intersection} contained in the reference \cite{Wellner-vdV} from before, which is the following bound on the VC-dimension of the \emph{union} of several VC-classes: 

\begin{proposition}[Lemma~1 in \cite{depersin2024robust}]\label{prop:union-vc} Suppose $\mathcal{F}_i, \, i=1,\dots,k$, are classes of Boolean functions such that $vc(\mathcal{F}_i) \le \nu, \, \forall \, i = 1,\dots,k.$ Then 
\[
vc(\F_1 \cup \F_2 \cup \cdots \cup \F_k) \le 2 \nu + \log(k).
\]
\end{proposition}

As a direct application of Propositions \ref{prop:vc-quantif-lln} and \ref{prop:intersection} to the case of half-spaces, we obtain the following Lemma: 

\begin{lemma}
\label{lem:VC_half-spaces}
Let $\mathcal{F}$ be a class of indicator functions $f:\R^n\rightarrow \{0,1\}$ of intersection of $k$ half-spaces, for some positive integer $k$. If $X_1,\ldots,X_m$ are i.i.d.~copies of a random vector $X \in \R^n$ then there exists a constant $C(k)$ for which
\begin{equation*}
\mathbb{E}\sup_{f\in \mathcal{F}}\left|\frac{1}{m}\sum_{i=1}^m f(X_i)-\mathbb{E}f(X)\right|\le C(k)\sqrt{\frac{n}{m}}.
\end{equation*}
\end{lemma}
\begin{proof}
By Proposition \ref{prop:vc-quantif-lln}, we know that 
\begin{equation*}
\mathbb{E}\sup_{f\in \mathcal{F}}\left|\frac{1}{m}\sum_{i=1}^m f(X_i)-\mathbb{E}f(X)\right|\le C\sqrt{\frac{vc(\mathcal{F})}{m}}.
\end{equation*}
All that remains is to estimate the VC-dimension of $\mathcal{F}$.  It is well-known that the VC-dimension of indicator functions of half-spaces in $\R^n$ is $n+1$.  Thus, the result follows by iteratively applying  Proposition \ref{prop:intersection} $\lceil \log_2(k)\rceil$ times. 
\end{proof}
\begin{remark}
We will need to apply the result for $k$ constant in the proof of Theorem \ref{thm:main-vc}, so the correct dependence with respect to $k$ only affects a multiplicative absolute constant in our final estimate, and may be neglected. 
\end{remark}

\section{Declipping and unlimited sampling for general distributions}\label{sec:general-results}

In this section, we shall prove the first four main results of this manuscript, dealing mainly with the questions of \emph{general} conditions such that random vectors satisfy stable recovery in the contexts of declipping and unlimited sampling. 

\subsection{Declipping for general distributions} We start the discussion of this section with a proof of our first main result for general distributions.

\begin{proof}[Proof of Theorem \ref{General distributions}] We divide our proof into two main cases: 

\vspace{1mm}

\noindent\textbf{Case 1:} Suppose that 
\begin{equation}\label{eq:prob-strips} 
\inf_{u,v \in \S^{n-1}} \mathbb{P}(|\langle X,u\rangle| \le \lambda, |\langle X,v\rangle| \le \lambda) \ge c'(\lambda). 
\end{equation} 
That is, the unclipped vectors are ``non-negligible''. Let then $$\mathcal{E}^1(u,v) = \left\{ |\langle X,u \rangle| \le \lambda \, \cap \, |\langle X,v\rangle| \le \lambda \right\} \cap \left\{ \left| \left\langle X, \frac{u-v}{\|u-v\|}\right\rangle \right| \ge \theta \right\}.$$ Notice that 
\begin{equation*}
\begin{split}
&\inf_{u,v \in\S^{n-1}}\frac{1}{m}\sum_{j=1}^m \frac{1}{\|u-v\|_{2}^2}\left|\Phi_{\lambda}(\langle X_j,u\rangle) - \Phi_{\lambda}(\langle X_j,v\rangle)\right|^2 \\
&\ge \theta^2 \inf_{u,v \in\S^{n-1}}\mathbb{P}(\mathcal{E}^1(u,v)) - \theta^2 \sup_{u,v\in\S^{n-1}}\left(\mathbb{P}(\mathcal{E}^1(u,v))-\frac{1}{m}\sum_{j=1}^m\mathds{1}(\mathcal{E}_j^1(u,v))\right).
\end{split}
\end{equation*}
Since we assume that $X$ satisfies a small ball condition such as \eqref{eq:small-ball-og}, we can guarantee that, by making $\theta > 0$ sufficiently small, we have, for any $u,v \in \B_{\R^n}$,
\[
\mathbb{P}(\mathcal{E}^1(u,v)) \ge \frac{1}{2}\mathbb{P}\left(|\langle X,u\rangle| \le \lambda, |\langle X,v\rangle| \le \lambda \right).
\]
It then suffices to note that the class 
\[
\mathcal{F}_{\lambda,\theta} = \left\{ \mathds{1}\left( |\langle x,u\rangle| \le \lambda \, \cap |\langle x, v \rangle| \le \lambda \, \cap |\langle x,w \rangle| \ge \theta\right), \, \, u,v \in \B_{\R^n}, \, w \in \S^{n-1} \right\}
\]
has VC-dimension $vc(\mathcal{F}_{\lambda,\theta}) \lesssim n$, as a consequence of Proposition \ref{prop:intersection} applied twice. It then follows from Proposition \ref{prop:bounded-differences} that 
\begin{align*} 
 & \sup_{u,v\in\S^{n-1}}\left(\mathbb{P}(\mathcal{E}^1(u,v))-\frac{1}{m}\sum_{j=1}^m\mathds{1}(\mathcal{E}_j^1(u,v))\right)  \cr 
\le & \, \mathbb{E} \sup_{u,v\in\S^{n-1}}\left(\mathbb{P}(\mathcal{E}^1(u,v))-\frac{1}{m}\sum_{j=1}^m\mathds{1}(\mathcal{E}_j^1(u,v))\right) + t 
\end{align*} 
with probability at least $1 - e^{-2t^2 m}$. By Lemma \ref{lem:VC_half-spaces}, we have that 
\[
\mathbb{E} \sup_{u,v\in\S^{n-1}}\left(\mathbb{P}(\mathcal{E}^1(u,v))-\frac{1}{m}\sum_{j=1}^m\mathds{1}(\mathcal{E}_j^1(u,v))\right) \lesssim \sqrt{\frac{n}{m}}.
\]
By choosing $t \sim_{\lambda} \sqrt{\frac{n}{m}}$ and $m \ge \tilde{c}(\lambda) n,$ where $\tilde{c}(\lambda)$ depends only on \eqref{eq:prob-strips}, we get that
\[
\inf_{u,v \in\S^{n-1}}\frac{1}{m}\sum_{j=1}^m \frac{1}{\|u-v\|_{2}^2}\left|\Phi_{\lambda}(\langle X_j,u\rangle) - \Phi_{\lambda}(\langle X_j,v\rangle)\right|^2 \ge c'(\lambda)^2,
\]
with probability at least $1 - e^{- \frac{m}{\tilde{c}(\lambda)}},$ as desired. 

\vspace{2mm}

\noindent \textbf{Case 2:} Suppose now, on the other hand, that \eqref{eq:prob-strips} does not admit a uniform lower bound and fix a  $\delta > 0.$ Consider the following events
\[
\mathcal{E}^3(u,v) = \left( |\langle X,u\rangle| \le \delta, |\langle X,v\rangle| > \lambda\right),
\]
\[
\mathcal{E}^4(u,v) = \left( |\langle X,u\rangle| \le \delta, |\langle X,v\rangle| \le \lambda\right),
\]
which satisfy that
$$\mathbb{P}(\mathcal{E}^3(u,v) \cup \mathcal{E}^4(u,v)) = \mathbb{P}(|\langle X,u\rangle| \le \delta).$$
Fix now $\delta = \frac{\lambda}{2}.$ Under the hypothesis of \eqref{eq:reverse-small-ball}, we conclude that there is $\tilde{c}(\lambda)$ such that 
\[
\inf_{u \in \S^{n-1}} \mathbb{P}(|\langle X,u\rangle| \le \lambda/2) \ge \tilde{c}(\lambda).
\]
We then divide into cases once more. 

\vspace{2mm}

\noindent \textit{Case 2.1:} $(u,v) \in \left\{ (u_0,v_0) \colon \mathbb{P}(\mathcal{E}^4(u_0,v_0)) \ge \frac{\tilde{c}(\lambda)}{2}\right\}.$ \medskip

In this case, we simply note that the exact same strategy as devised above for \textbf{Case 1} may be applied verbatim, and hence we need not deal with this case in particular. 

\vspace{2mm}

\noindent\textit{Case 2.2:} $(u,v) \in \left\{ (u_0,v_0) \colon \mathbb{P}(\mathcal{E}^4(u_0,v_0)) \le \frac{\tilde{c}(\lambda)}{2}\right\}.$ \medskip

Since $\mathcal{E}^3(u,v)$ and $\mathcal{E}^4(u,v)$ are disjoint, if the probability of, say, $\mathcal{E}^4(u,v)$ is less than $\tilde{c}(\lambda)/2,$ then 
\[
\mathbb{P}(\mathcal{E}^3(u,v)) \ge \frac{\tilde{c}(\lambda)}{2}. 
\]
In this case, instead of restricting to the set of unclipped samples, we simply restrict to the set where one of the measurements is clipped and the other one is far from being clipped. Indeed, we have
\begin{equation*}
\begin{split}
&\frac{1}{m}\sum_{j=1}^m \frac{1}{\|u-v\|_{2}^2}\left|\Phi_{\lambda}(\langle X_i,u\rangle) - \Phi_{\lambda}(\langle X_i,v\rangle)\right|^2 \\
&\ge \frac{1}{\|u-v\|_2^2} \frac{\lambda^2}{4m} \sum_{j=1}^m \mathds{1}\left(|\langle X_i,u\rangle| > \lambda \cap |\langle X_i,v\rangle|\le \lambda/2 \right) \\ 
& \ge \frac{\lambda^2}{16m} \sum_{j=1}^m \mathds{1}\left(|\langle X_i,u\rangle| > \lambda \cap |\langle X_i,v\rangle|\le \lambda/2 \right). 
\end{split}
\end{equation*}
We then employ Lemma \ref{lem:VC_half-spaces} once more. This yields that the difference between the empirical process above and its expected value -- which is $\mathbb{P}(\mathcal{E}^3(u,v)) \ge \frac{\tilde{c}(\lambda)}{2}$ -- is at most $O(\sqrt{n/m})$ with probability $1 - e^{-c n}$, and we conclude the result also in this case. 
\end{proof}

\begin{remark}

Condition \eqref{eq:reverse-small-ball} states, as previously implied, a ``reverse'' small ball property: while the probability that $|\langle X,u\rangle|$ is not very small gets (uniformly) progressively closer to $1$, it also stays uniformly away from $1$. In other words, this is a sort of assertion on the roundness of the probability distribution of the vector $X$. 

While this, at first, seems to be an overly restrictive assumption, it is not hard to believe that such a condition should be in some sense \emph{necessary} in case declipping does hold. Indeed, if not, then \eqref{eq:reverse-small-ball} states that, for some specific vectors, most measurement will end up being clipped, and as such too much information should be lost. We refer the reader to Remark \ref{rmk:rsb-needed} below for a formal proof of the fact that \eqref{eq:reverse-small-ball} is indeed \emph{necessary} for stable declipping. We also refer the reader to Remark \ref{rmk:ssb-sharp?} for a discussion on the condition \eqref{eq:small-ball-og}. 
\end{remark} 

\subsection{Exact characterization of declipping for some \texorpdfstring{$\lambda$}{}} We move now to the proof of Theorem \ref{thm:general-some-lambda}. 

\begin{proof}[Proof of Theorem \ref{thm:general-some-lambda}]
Suppose first that assertion (I) in Theorem \ref{thm:general-some-lambda} holds. By \cite[Theorem 2.1]{Stablephaseretrieval} we have a dichotomy: Either  $E$ contains an almost disjoint normalized sequence, or the $L^1$ and $L^p$ norms are equivalent on $E$. Hence, if assertion (II) fails, we can find a sequence $(f_n)\subseteq S_E$ together with a sequence $\{d_n\}_n$ of pairwise disjointly supported functions such that $\|f_n-d_n\|_{L^p(\mu)}\to 0.$ In particular, as $\mu(\text{supp}(d_n)) \to 0$, it follows that $f_n$ converges to zero in measure. This implies that for each $\lambda>0$, $\|\Phi_{\lambda}(f_n)\|_{L^p}\to 0$, contradicting (I).
\medskip

Now, suppose that assertion (II) in Theorem \ref{thm:general-some-lambda} holds, but assertion (I) of Theorem \ref{thm:general-some-lambda} fails. Then for each $n\in \mathbb{N}$ we may find $f_{n}, g_{n}\in B_E$ with
\begin{equation}\label{fails clip for all n}
\|f_{n}-g_{n}\|_{L^p}>n \|\Phi_n(f)-\Phi_n(g)\|_{L^p}.
\end{equation}
Let $\Omega_n$ be the event where $\{|f_n|\  \text{and} \ |g_n|<n\}$. Since $f_n,g_n\in B_E$, we have $\mu(\Omega_n)\geq 1-\frac{2}{n^p}.$
On the other hand, by \eqref{fails clip for all n}, we have 
\begin{equation*}
\|f_n|_{\Omega_n}-g_n|_{\Omega_n}\|_{L^p}<\frac{\|f_n-g_n\|_{L^p}}{n}.
\end{equation*}
Hence,
\begin{equation*}
\|f_n-g_n\|^p_{L^p}\left(1-\frac{1}{n^p}\right)\leq \|(f_n-g_n)|_{\Omega_n^c}\|_{L^p}^p.
\end{equation*}
This shows that $f_n-g_n$ accumulates most of its mass on a set of very small measure. In particular,
$$\|f_n-g_n\|_{L^1}=\int_{\Omega_n} |\Phi_n(f)-\Phi_n(g)|+\int_{\Omega_n^c} |f_n-g_n|$$

$$< \frac{\|f_n-g_n\|_{L^p}}{n}+\|f_n-g_n\|_{L^p}\mu(\Omega_n^c)^{1/p'}$$

$$\leq \left(\frac{1}{n}+\left(\frac{2}{n^p}\right)^{\frac{1}{p'}}\right)\|f_n-g_n\|_{L^p}.$$
This shows that assertion (II) fails, a contradiction which concludes the proof. 
\end{proof}

\begin{remark}\label{rmk:ssb-sharp?} We have characterized exactly when a distribution $\mu$ satisfies declipping for \emph{some} $\lambda >0.$ However, we note that the strong small ball condition \eqref{eq:small-ball-og} is \emph{not} necessary in order for declipping to hold in the continuous setting for \emph{all} $\lambda>0$: take, for instance, $\mu$ to be a probability distribution that gives $1/2$ weight to a single fixed point $p_0$ in the sphere $\sqrt{n} \S^{n-1}$, and the other half to the rest of the sphere, attributed in a uniform way. Take $X$ to be a random variable distributed according to this distribution. This random vector clearly does not satisfy \eqref{eq:small-ball-og}, but since a random vector $Y$ which is uniformly distributed on the sphere $\sqrt{n} \S^{n-1}$ satisfies that, for all $u,v \in \B_{\R^n}$,
\begin{equation}\label{eq:continuous-dec}
    \mathbb{E} \left( \Phi_\lambda(\langle Y,u \rangle) - \Phi_\lambda (\langle Y,v\rangle)\right)^2 \ge c(\lambda)\|u-v\|_2^2,
\end{equation}
for some constant $c(\lambda)>0,$ then $X$ must satisfy the same kind of inequality, showing the asserted fact that \eqref{eq:small-ball-og} is not needed in order for declipping to hold. 

One may wonder, however, if a weaker version is needed in order for declipping to hold: namely, does it hold that if a probability measure $\mu$ is such that any random vector $X$  distributed according to $\mu$ satisfies \eqref{eq:continuous-dec}, then we may write $\mu = \mu_1 + \mu_2$, where $\mu_1$ and $\mu_2$ have disjoint supports, and at least one of them satisfies \eqref{eq:small-ball-og}?

Perhaps surprisingly, the next result shows that even this weaker form is \emph{not} needed in order for \eqref{eq:continuous-dec} to hold:
\begin{proposition}\label{prop:double-sphere-ex} Let $\S^{n-2}_i:=\S^{n-1}\cap \{x_i =0\}$. Take $\mu = \frac{c_n}{2} \left( \mathcal{H}^{n-2}|_{\sqrt{n}\S^{n-2}_1} +  \mathcal{H}^{n-2}|_{\sqrt{n}\S^{n-2}_2} \right),$ where the constant $c_n>0$ is chosen so that $\mu$ is a probability measure. If $X$ is a random vector distributed according to $\mu,$ then it satisfies \eqref{eq:continuous-dec}, although $X$ does not satisfy \eqref{eq:small-ball-og}. 
\end{proposition}

\begin{proof} Let $u,v \in \B_{\R^n}$ be fixed. Let then $u_i,v_i$ denote the respective projections of $u,v$ onto the planes $\{x_i = 0\}, i=1,2.$ We first note that, directly by the way we have defined the spaces we are working on, it follows that 
\[
\|u_1 - v_1 \|_2^2 + \|u_2 - v_2 \|_2^2 \ge \| u-v\|_2^2. 
\]
Now, let $Y_1,Y_2$ denote random variables, each distributed uniformly on $\sqrt{n}\S^{n-1}_i$, respectively. Since 
\[
\langle Y_i,u\rangle =\langle Y_i, u_i\rangle, \quad \langle Y_i,v\rangle =\langle Y_i,v_i\rangle,
\]
it follows that 
\[
\mathbb{E} \left( \Phi_\lambda(\langle Y_i,u \rangle) - \Phi_\lambda (\langle Y_i,v\rangle)\right)^2 \ge c(\lambda)\|u_i-v_i\|_2^2,
\]
where we used the fact that the uniform distribution on the sphere of radius $\sqrt{n}$ satisfies \eqref{eq:continuous-dec}. Now, since we may write 
$$X = \varepsilon_0 \cdot Y_1 + (1-\varepsilon_0)\cdot Y_2,$$ 
where $\varepsilon_0$ is a Bernoulli $1/2$-random variable independent of $Y_i,i=1,2,$ 
it follows that 
\begin{align*}
\mathbb{E} \left( \Phi_\lambda(\langle X,u \rangle) - \Phi_\lambda (\langle X,v\rangle)\right)^2 & = \frac{1}{2} \left(\mathbb{E} \left( \Phi_\lambda(\langle Y_1,u \rangle) - \Phi_\lambda (\langle Y_1,v\rangle)\right)^2 + \mathbb{E} \left( \Phi_\lambda(\langle Y_2,u \rangle) - \Phi_\lambda (\langle Y_2,v\rangle)\right)^2 \right) \cr
& \ge \frac{c(\lambda)}{2} \left( \|u_1 - v_1 \|_2^2 + \|u_2 - v_2 \|_2^2 \right) \ge \frac{c(\lambda)}{2} \|u-v\|_2^2.
\end{align*}
Since it is easy to see that $X$ does not satisfy \eqref{eq:small-ball-og}, this finishes the proof of the proposition. 
\end{proof}

The question of what condition characterizes doing stable $\lambda$-declipping for all $\lambda$ seems to be an interesting one. On the one hand, the reverse small ball condition \eqref{eq:reverse-small-ball} is necessary, and Theorem \ref{thm:general-some-lambda} implies that the weak small ball (rewritten in the statement of that result as the equivalence of $L^1$ and $L^2$ norms) is also necessary. On the other hand, by virtue of the examples above, a complete characterization seems to be more elusive than the current techniques, being a problem we wish to investigate further in future work. 
\end{remark}


\subsection{General declipping in the sparse context}

\begin{proof}[Proof of Theorem \ref{thm:sparse-gen-dist}]

We follow exactly the same method of proof as in Theorem \ref{General distributions}. Indeed, by repeating the first few steps, we obtain in the exact same manner that

\begin{equation*}
\begin{split}
&\inf_{u,v \in T_s}\frac{1}{m}\sum_{j=1}^m \frac{1}{\|u-v\|_{2}^2}\left|\Phi_{\lambda}(\langle X_j,u\rangle) - \Phi_{\lambda}(\langle X_j,v\rangle)\right|^2 \\
&\ge L^2 \inf_{u,v \in T_s}\mathbb{P}(\mathcal{E}(u,v)) - L^2 \sup_{u,v\in T_s}\left(\mathbb{P}(\mathcal{E}(u,v))-\frac{1}{m}\sum_{j=1}^m\mathds{1}(\mathcal{E}_j(u,v))\right),
\end{split}
\end{equation*}
where the event $\mathcal{E}(u,v)$ is defined in \eqref{eq:event-E(u,v)}. Since we have already proved a lower bound on the probability of that event for arbitrary $u,v \in \B_{\R^n}$, we only need to estimate the term 
\[
\sup_{u,v\in T_s}\left(\mathbb{P}(\mathcal{E}(u,v))-\frac{1}{m}\sum_{j=1}^m\mathds{1}(\mathcal{E}_j(u,v))\right).
\]
We again repeat the overarching structure of the proof of Theorem \ref{General distributions}: again by Proposition \ref{prop:bounded-differences}, we obtain for all $t>0$ with probability at least $1-e^{-2t^2m}$ that
\begin{equation*}
\begin{split}
&\sup_{u,v\in T_s}\left(\mathbb{P}(\mathcal{E}(u,v))-\frac{1}{m}\sum_{j=1}^m\mathds{1}(\mathcal{E}_j(u,v))\right) \\
&\le \mathbb{E} \sup_{u,v\in T_s}\left(\mathbb{P}(\mathcal{E}(u,v))-\frac{1}{m}\sum_{j=1}^m\mathds{1}(\mathcal{E}_j(u,v))\right) + t.\\
\end{split}
\end{equation*}
Our main divergence point is the usage of Proposition \ref{prop:union-vc}: indeed, by using that result, we get that any family of the form 
\[
\F_\lessgtr = \left\{ \mathds{1}\left(|\langle x,u\rangle| \lessgtr \theta \right) \colon u \in T_s \right\}
\]
satisfies $vc(\F_\lessgtr) \lesssim s \log(e \cdot n/s).$ This follows, for instance, by combining Proposition \ref{prop:union-vc} with the fact that $T_s$ is the union of ${n \choose s}$ Euclidean balls of dimension $s$. By using Proposition \ref{prop:intersection} twice, we obtain that the VC-dimension of the class 
\[
\tilde{\F}_{\lambda,\theta}\coloneqq\left\{ \mathds{1}\left( |\langle x,u\rangle| \le \lambda \, \cap |\langle x, v \rangle| \le \lambda \, \cap |\langle x,w \rangle| \ge \theta\right), \, \, u,v \in \B_{\R^n}\cap T_s, \, w \in \S^{n-1}\cap T_s \right\},
\]
satisfies $vc(\tilde{\F}_{\lambda,\theta}) \lesssim s \log(e \cdot n/s).$ Applying now Proposition \ref{prop:vc-quantif-lln}, we have 
\[
\mathbb{E} \sup_{u,v\in T_s}\left(\mathbb{P}(\mathcal{E}(u,v))-\frac{1}{m}\sum_{j=1}^m\mathds{1}(\mathcal{E}_j(u,v))\right) \lesssim \sqrt{\frac{s \log\left( \frac{e \cdot n}{s}\right)}{m}}.
\]
Now, the remainder of the proof of Theorem \ref{General distributions} can be applied verbatim to conclude the argument.
\end{proof}

\subsection{Unlimited sampling for general distributions} We now proceed to prove Theorem \ref{thm:mod-sharp-try}.

\begin{proof}[Proof of Theorem \ref{thm:mod-sharp-try}] We split the proof of the desired result into cases. \medskip

\noindent\textbf{Case 1:} $\|u-v\|_2 \le \lambda. $ In this case, we simply note that, if $s,t$ are such that $|s-t|\le \lambda$, then $|\mathcal{M}_\lambda(s) - \mathcal{M}_\lambda(t)| \ge |s-t|.$ Restricting to the event 
$$\left\{ a \|u-v\|_2 \le \left|\left\langle X,u-v\right\rangle\right| \le \|u-v\|_2 \right\},$$
we have 
\begin{align*}
\frac{1}{m} \sum_{i=1}^m \frac{|\mathcal{M}_\lambda(\langle X_i,u\rangle) - \mathcal{M}_\lambda(\langle X_i,v\rangle)|^2}{\|u-v\|_2^2} & \ge \frac{a^2}{m} \sum_{i=1}^m \mathds{1}\left( a \le \left| \left\langle X_i, \frac{u-v}{\|u-v\|_2} \right\rangle \right| \le 1 \right).
\end{align*}
Now, note that, by condition \eqref{eq:small-strip}, there is $a>0$ such that 
\[
\mathbb{P} \left( a \le |\langle X,w\rangle\right| \le 1) \ge c > 0. 
\]
Since, in addition, the class of Boolean functions 
\[
\F = \left\{ \mathds{1}(a \le |\langle x,w\rangle| \le 1) \colon w \in \S^{n-1}\right\}
\]
has, by virtue of Proposition \ref{prop:intersection}, VC-dimension bounded by $n$, it follows from Proposition \ref{prop:vc-quantif-lln} that 
\[
\mathbb{E} \sup_{u,v \in \B_{\R^n}} \left[ \frac{1}{m} \sum_{i=1}^m \mathds{1}\left( a \le \left| \left\langle X_i, \frac{u-v}{\|u-v\|_2} \right\rangle \right| \le 1 \right) - \mathbb{P}\left( a \le \left| \left\langle X,\frac{u-v}{\|u-v\|_2}\right\rangle \right| \le 1\right) \right] \le C \sqrt{\frac{n}{m}}.  
\]
Arguing again with Proposition \ref{prop:bounded-differences}, as in the proof of Theorem \ref{General distributions}, we have that 
\[ 
\frac{1}{m} \sum_{i=1}^m \mathds{1}\left( a \le \left| \left\langle X_i, \frac{u-v}{\|u-v\|_2} \right\rangle \right| \le 1 \right) \ge c - C\sqrt{\frac{n}{m}},
\]
with probability at least $1 - e^{-c t^2 m}.$ Setting $m \ge C n$ for $C$ an absolute, sufficiently large, constant  finishes the proof in this case. \medskip

\noindent\textbf{Case 2:} $\|u-v\|_2 \ge \lambda.$ We now simply note that, for the sets $\Omega_a(\delta)$ as defined in the statement of Theorem \ref{thm:mod-sharp-try}, we have that $|\mathcal{M}_{\lambda}(s) - \mathcal{M}_\lambda(t)| \ge a\lambda$ holds if $s-t \in \Omega_a(\lambda).$ Hence, translating this as a lower bound, we have 
\begin{equation}\label{eq:lower-bound-mod} 
\frac{1}{m} \sum_{i=1}^m \frac{|\mathcal{M}_\lambda(\langle X_i,u\rangle) - \mathcal{M}_\lambda(\langle X_i,v\rangle)|^2}{\|u-v\|_2^2} \ge \frac{(a\lambda)^2}{\|u-v\|_2^2} \frac{1}{m} \sum_{i=1}^m \mathds{1}\left( \langle X_i, u-v\rangle \in \Omega_a(\lambda)\right).
\end{equation} 
Now, note that 
\[
\langle X,u-v \rangle \in \Omega_a(\lambda) \, \text{ follows if }\, \langle X, (u-v)/\|u-v\|_2\rangle \in \Omega_a(\lambda/\|u-v\|_2).
\]
From condition \eqref{eq:mod-cond}, we obtain that 
\[
\inf_{u,v \in \S^{n-1}} \mathbb{P} \left( \langle X,u-v\rangle \in \Omega_a(\lambda)\right) \ge \inf_{\delta \ge \lambda/4} \tilde{c}(\delta) =: \theta(\lambda)> 0.
\]
Let now $C_\lambda > 0$ be such that we have 
\[
\mathbb{P}(|\langle X,u-v\rangle| \ge C_\lambda) \le \frac{\theta(\lambda)}{2},
\]
for all $u,v \in \B_{\R^n}$ such that $\|u-v\|_2 \ge \lambda.$ It follows that 
\[
\mathbb{P}\left( \langle X,u-v\rangle \, \in \, \Omega_a(\lambda) \cap [-C_\lambda,C_\lambda]\right) \ge \frac{\theta(\lambda)}{2}.
\]
On the other hand, it follows from Proposition \ref{prop:intersection} that the family 
\[
\F_{a,\lambda} \coloneqq \left\{ \mathds{1}\left( \langle X,u-v\rangle \in \Omega_a(\lambda) \cap [-C_\lambda,C_\lambda]\right)\colon u,v \in \B_{\R^n}\right\},
\]
being composed of indicator functions of a finite, fixed number of classes, each of which having VC-dimension bounded by $n$, satisfies $vc(\F_{a,\lambda}) \le C(\lambda)n.$ Hence, 
\[
\mathbb{E} \sup_{u,v \in \B_{\R^n}} \left[ \frac{1}{m} \sum_{i=1}^m \mathds{1}\left( \langle X_i, u-v \rangle \in \Omega_a(\lambda) \cap [-C_\lambda,C_\lambda] \right) - \mathbb{P}\left( \langle X,u-v\rangle \, \in \, \Omega_a(\lambda) \cap [-C_\lambda,C_\lambda]\right)\right] \]
\[
\le C(\lambda) \sqrt{\frac{n}{m}}. 
\]
Again, employing the usual device of Proposition \ref{prop:bounded-differences}, we conclude that 
\begin{align*}
\frac{1}{m} \sum_{i=1}^m \mathds{1}\left( \langle X_i,u-v\rangle \in \Omega_a(\lambda)\right) & \ge \frac{1}{m} \sum_{i=1}^m \mathds{1}\left( \langle X_i,u-v\rangle \in \Omega_a(\lambda) \cap [-C_\lambda,C_\lambda]\right) \cr 
    & \ge \mathbb{P} \left( \langle X,u-v\rangle \in \Omega_a(\lambda) \cap [-C_\lambda,C_\lambda]\right) - C(\lambda) \sqrt{\frac{n}{m}}-t,
\end{align*}
with probability at least $1 - e^{-ct^2 m}.$ By taking $t$ to be a small constant depending on $\lambda$ and $m \ge C(\lambda)n$ for some $C(\lambda)> 0$, we conclude that the right-hand side of \eqref{eq:lower-bound-mod} is bounded from below by $c(\lambda) > 0$ with probability at least $1- e^{-c(\lambda)m}$, as long as $m \ge C(\lambda) n.$ This concludes the proof. 
\end{proof}

\section{Sharp declipping and unlimited sampling for specific distributions}\label{sec:declipping}

In this section, our focus will be mostly on results dealing with \emph{sharp} asymptotic lower bounds for both declipping and unlimited sampling in specific cases -- such as when the random vector under consideration satisfies additional geometric constraints, like being uniformly distributed on a certain sphere.

\subsection{Sharp declipping for uniformly spherically distributed vectors} We start with the proof of Theorem \ref{thm:main-vc}. In order to do so, we need the following auxiliary result: 

\begin{lemma}\label{L:Clam}
    Let $\lambda,L\in\R$ be positive numbers with $4L\le \lambda \le 1$. If $u,v\in \B_{\R^n}$ and $x\in\R^n$ satisfy 
$$
    |\langle x,u\rangle|\le \frac{\lambda}{2} \hspace{.2cm}\textrm{ and }\hspace{.2cm} \left|\left\langle x,u-v\right\rangle\right| \ge L \|u-v\|_2,
$$
then
$$
\left|\Phi_{\lambda}(\langle x,u\rangle) - \Phi_{\lambda}(\langle x,v\rangle)\right| \ge L\|u-v\|_{2}.
 $$   
\end{lemma}

\begin{proof}
We first consider the case when $|\langle x,v\rangle|\leq \lambda$. It follows that
$$
\left|\Phi_{\lambda}(\langle x,u\rangle) - \Phi_{\lambda}(\langle x,v\rangle)\right|=|\langle x,u-v\rangle|\geq L \|u-v\|.
$$
We now assume that $|\langle x,v\rangle|> \lambda$.  Then,
\begin{align*}
\left|\Phi_{\lambda}(\langle x,u\rangle) - \Phi_{\lambda}(\langle x,v\rangle)\right|&\geq \lambda-\lambda/2 \\
&\geq L\|u-v\|\hspace{.5cm}\textrm{ as $\|u-v\|\leq 2$ and  $L\leq \lambda/4$,}
\end{align*}
which finishes the proof.
\end{proof}

We are now ready to prove Theorem \ref{thm:main-vc}. 

\begin{proof}[Proof of Theorem \ref{thm:main-vc}]
We consider $0<\lambda\le 1/2$ and
fix a constant $L>0$ with $\lambda>4L$ to be defined later. By Lemma \ref{L:Clam}, we have that
\begin{align}\label{eq:lower-decl-sharp}
\frac{1}{m}\sum_{i=1}^m \frac{1}{\|u-v\|_{2}^2}\left|\Phi_{\lambda}(\langle X_i,u\rangle) - \Phi_{\lambda}(\langle X_i,v\rangle)\right|^2 \ge \frac{L^2}{m} \sum_{i=1}^m \mathds{1}\left(|\langle X_i,u\rangle|\le \frac{\lambda}{2} \cap \left|\left\langle X_i,\frac{u-v}{\|u-v\|_2}\right\rangle\right| \ge L\right).
\end{align}
We define $\mathcal{E}(u,v)$ to be the event  
\begin{equation}\label{eq:event-E(u,v)}
\mathcal{E}(u,v) \coloneqq \left\{|\langle X,u\rangle|\leq \frac{\lambda}{2} \cap \left|\left\langle X,\frac{u-v}{\|u-v\|_2}\right\rangle\right|\ge L\right\},
\end{equation} 
and let the events $\mathcal{E}_i(u,v)$ (corresponding to $X_i$) be defined analogously. \medskip

We start with the (uniform) small ball probability. To this end, note that as $\|u\|\leq 1$ and $X$ is isotropic, we have by Markov's inequality that $\mathbb{P}(|\langle X,u\rangle|\leq \sqrt{2})\geq 1/2$.  As $X$ is uniformly distributed on a sphere, for each $\alpha>0$, $\mathbb{P}(|\langle X,u\rangle|\le \alpha)$ corresponds to the surface area of a central strip of width $2\alpha$ on an $(n-1)$-sphere.  As the area of a strip divided by its width is comparable to the surface area of an $(n-2)$-sphere, 
 there exists universal constants $c_1>c_0>0$ so that 
$$c_1\alpha \mathbb{P}\big(|\langle X,u\rangle|\le \sqrt{2}\big)\geq\mathbb{P}\big(|\langle X,u\rangle|\le \alpha\big)\geq 2c_0 \alpha  \mathbb{P}\big(|\langle X,u\rangle|\le \sqrt{2}\big).
$$
Thus, we have that 
\begin{equation}\label{E:proportional}    
c_1\alpha \geq\mathbb{P}(|\langle X,u\rangle|\le \alpha)\geq c_0 \alpha.  
\end{equation}
As \eqref{E:proportional} applies as well if $u$ is replaced with $(u-v)/\|u-v\|$, we have that

\begin{align*}
\mathbb{P}(\mathcal{E}(u,v))&\ge \mathbb{P}(\left|\langle X,u\rangle\right|\le \lambda/2)-\mathbb{P}\left(\left|\left\langle X,\frac{u-v}{\|u-v\|}\right\rangle\right|< L\right)\\
&\ge c_0\lambda-c_1 L.
\end{align*}

Let now $\mathcal{N} \coloneqq\{w_i\}_{i \in I}$ denote a $\lambda^2$-net of the unit ball $\B_{\R^n}.$ It follows that
$|\mathcal{N}| \le \left( \frac{C}{\lambda^2}\right)^{n}.$ Now, we need a statement that bounds the right-hand side of \eqref{eq:lower-decl-sharp} by its expected value for each element of our net with high probability. One possible way to do this would be through the usage of Proposition \ref{prop:bounded-differences}, which would imply that, for each $w_i \in \mathcal{N},$ with probability at least $1 - e^{- c \lambda^2 m}$, we have 
\begin{equation}\label{eq:lower-prob-high-each}
\frac{1}{m} \sum_{i=1}^m \mathds{1}(\mathcal{E}_i(u,v)) \ge \frac{1}{2} \cdot \mathbb{P} (\mathcal{E}(u,v)). 
\end{equation} 
This bound is, however, rather \emph{crude}: the specific Bernoulli structure of our random variables allows us to use the better-tailored \emph{Chernoff bound} for Bernoulli random variables in order to bound the probability of the event where 
\[
\left\{\frac{1}{m} \sum_{i=1}^m \mathds{1}(\mathcal{E}_i(u,v)) \le \frac{1}{2} \cdot \mathbb{P} (\mathcal{E}(u,v))\right\}
\]
by at most $e^{-c \lambda m}$. Hence, we finally conclude that, for each $w_i \in \mathcal{N},$ we have \eqref{eq:lower-prob-high-each} with probability at least $1 - e^{-c \lambda \cdot m}$.

\vspace{2mm}

Now, in order to further bound \eqref{eq:lower-decl-sharp} from below, we note that the intersection of the event 
\[
\mathcal{A}(u_k,w_j) \coloneqq\left\{ |\langle X,u_k\rangle| \le \frac{\lambda}{4} \, \cap \left| \langle X,w_j \rangle \right| > 2L \right\}
\]
with the event 
\[
\mathcal{B}(u,u_k,w,w_j) \coloneqq\left\{ |\langle X, u_k-u\rangle| \le \frac{\lambda}{10} \, \cap |\langle X, w_j-w\rangle| \le \frac{\lambda}{10}\right\}
\]
is contained in $\mathcal{E}(u,v)$, for $w = \frac{u-v}{\|u-v\|_2}$. Hence, by picking $u_k$ to be an element of $\mathcal{N}$ such that $\|u_k - u\|_2 \le c \lambda^2$ and $w_j \in \mathcal{N}$ such that $\|w_j - w\|_2 \le c \lambda^2$, we obtain 
\begin{align}\label{eq:second-lower-declp-sharp}
\frac{1}{m} \sum_{i=1}^m \mathds{1}(\mathcal{E}_i(u,v)) & \ge \frac{1}{m} \sum_{i=1}^m \mathds{1} \left(\mathcal{A}_i(u_k,w_j) \, \cap \mathcal{B}_i(u,u_k,w,w_j)\right) \cr 
 & \ge \frac{1}{m} \left( \sum_{i=1}^m \mathds{1}\left( \mathcal{A}_i(u_k,w_j)\right) - \sum_{i=1}^m \mathds{1}\left(\mathcal{B}_i^{c} (u,u_k,w,w_j)\right) \right). 
\end{align}
We then use a simple union bound: with probability at least 
$$
1 - |\mathcal{N}|^2 e^{-c \lambda m} \ge 1 - e^{c \left( n \log\left( \frac{1}{\lambda}\right) - \lambda m\right)},
$$
the first summand on the right-hand side of \eqref{eq:second-lower-declp-sharp} is bounded from below by $ c \cdot \lambda,$ for some universal, positive constant $c > 0$. On the other hand, for the second term, it suffices to notice that 
\begin{align}\label{eq:upper-bound-B}
\sum_{i=1}^m \mathds{1}\left(\mathcal{B}_i^{c} (u,u_k,w,w_j)\right) & \le \sum_{i=1}^m \mathds{1}\left( |\langle X_i, u_k - u \rangle| \ge \frac{\lambda}{10} \right) + \sum_{i=1}^m \mathds{1} \left( |\langle X_i, w_j - w\rangle| \ge \frac{\lambda}{10} \right) \cr 
                    & \le 100 \left( \sum_{i=1}^m \frac{|\langle X_i, u_k-u\rangle|^2}{\lambda^2} + \sum_{i=1}^m \frac{|\langle X_i,w_j-w\rangle|^2}{\lambda^2} \right).
\end{align}
We now invoke Proposition \ref{prop:vershynin-subgauss}. This, together with \eqref{eq:second-lower-declp-sharp} and \eqref{eq:upper-bound-B} implies that, with probability at least $1 - e^{-c \lambda m}$, we have 
\begin{align}
\frac{1}{m} \sum_{i=1}^m \mathds{1}(\mathcal{E}_i(u,v)) & \ge c \lambda - \frac{C}{\lambda^2}\left(\|u_k-u\|_2^2 + \|w_j-w\|_2^2\right) \cr 
    & \ge c \lambda - C \lambda^2 \ge \frac{c}{2} \lambda,
\end{align} 
for $\lambda > 0$ sufficiently small. This, together with \eqref{eq:lower-decl-sharp}, implies that 
\[
\frac{1}{m}\sum_{i=1}^m \frac{1}{\|u-v\|_{2}^2}\left|\Phi_{\lambda}(\langle X_i,u\rangle) - \Phi_{\lambda}(\langle X_i,v\rangle)\right|^2 \ge c \cdot\lambda^3
\]
with probability at least $1 - e^{- c \lambda m},$ as long as we have $m \gtrsim \lambda^{-1} n \log\left( \frac{1}{\lambda}\right),$ as desired. 
\end{proof}

\begin{remark} We note that, in general, a bound of the form $m \ge c \lambda^{-1} n$ is unavoidable if we wish for the general results in this manuscript to hold with probability at least $1-e^{-cn}$. Indeed, suppose that, with probability at least $1-e^{-cn}$, we have 
\[
\frac{1}{m}\sum_{j=1}^m \frac{1}{\|u-v\|_{2}^2}\left|\Phi_{\lambda}(\langle X_i,u\rangle) - \Phi_{\lambda}(\langle X_i,v\rangle)\right|^2 \ge c(\lambda),
\]
for some constant $c(\lambda) >0$ and all $u,v \in \B_{\R^n},$ where we assume the basic framework of Theorem \ref{General distributions}. Then, letting $v = (1-\varepsilon)u$ and making $\varepsilon \to 0$ as in the proof of Theorem \ref{thm:best-lam}, we should have that 
\begin{align*} 
\lambda^2 \frac{1}{m} \sum_{j=1}^m \mathds{1}_{\{|\langle X_i, u\rangle| \le \lambda\}} & \ge \frac{1}{m}\sum_{j=1}^m |\langle X_i, u\rangle|^2 \mathds{1}_{\{|\langle X_i,u\rangle| \le \lambda\} } \cr 
& \ge c(\lambda),
\end{align*}
that is, 
\begin{equation}\label{eq:Bernoulli-lower} 
\frac{1}{m} \sum_{j=1}^m \mathds{1}_{\{|\langle X_i, u\rangle| \le \lambda\}} \ge \frac{c(\lambda)}{\lambda^2}.
\end{equation} 
Now, since the variables $\left\{ \mathds{1}_{\{ |\langle X_i,u\rangle| \le \lambda\}} \right\}_{i=1}^m$ are actually i.i.d.~Bernoulli$(p)$ variables, where $p$ is of the order of $\lambda$, the probability that the left-hand side of \eqref{eq:Bernoulli-lower} is non-zero is at most $1-(1 - \lambda)^m.$ Hence, putting everything together, we obtain that 
\[
1-(1-\lambda)^m \ge 1 - e^{-cn},
\]
which, upon rearranging and using that $(1-\lambda)^m$ is of the order of $e^{-\lambda m}$ for $\lambda$ small, yields that $\lambda m \ge c n$, which is the claimed bound. 
\end{remark}

We now show that, in general, the bound achieved by Theorem \ref{thm:main-vc} is \emph{sharp}: 

{ 
\begin{proof}[Proof of Theorem \ref{thm:best-lam}]
Let $(x_j)_{j=1}^m$ be any fixed collection of vectors in $\sqrt{n}\S^{n-1}$.  For each $u\in \S^{n-1}$ we have that 
\begin{equation*}
\lim_{\varepsilon\rightarrow0}\frac{1}{m}\sum_{i=1}^m \frac{1}{\|u-(1-\varepsilon)u\|_{2}^2}\left|\Phi_{\lambda}(\langle u,x_j\rangle) - \Phi_{\lambda}(\langle (1-\varepsilon)u,x_j\rangle)\right|^2= \frac{1}{m}\sum_{|\langle u,x_j\rangle|\leq\lambda} \left|\langle u,x_j\rangle\right|^2.\end{equation*}
Thus, in order to prove Theorem \ref{thm:best-lam} it is sufficient to prove that there exists a universal constant $\beta>0$ so that there exists $u\in \S^{n-1}$ such that $\frac{1}{m}\sum_{|\langle u,x_j\rangle|\leq\lambda} \left|\langle u,x_j\rangle\right|^2<\beta \lambda^{3}$.  
This will be done by proving the claim that if $c_1$ is the constant given in Theorem \ref{thm:main-vc} and if $U$ is a random vector which is uniformly distributed in $\S^{n-1}$ then  
\begin{equation}\label{E:smaller}
\E \frac{1}{m}\sum_{|\langle U,x_j\rangle|\leq\lambda}|\langle U,x_j\rangle|^2\leq c_1\lambda^{3}. 
\end{equation}
Let 
 $U$ be a random vector which is uniformly distributed in $\S^{n-1}$ and let $x$ be a fixed vector in $\sqrt{n}\S^{n-1}$.  This scaling keeps \eqref{E:proportional} invariant, and thus we have for the same universal constants $c_1> c_0$ in the proof of Theorem \ref{thm:main-vc} that
 \begin{equation*}\label{E:proportional2}    
c_1\lambda \geq\mathbb{P}(|\langle U,x\rangle|\le \lambda)\geq c_0 \lambda.  
\end{equation*}
We now have that
\begin{align*}
    \E \frac{1}{m}\sum_{|\langle U,x_j\rangle|\leq\lambda}|\langle U,x_j\rangle|^2
    &\leq \E \frac{1}{m}\sum_{|\langle U,x_j\rangle|\leq\lambda} \lambda^2\\
&\leq \mathbb{P}\,(|\langle U,x\rangle|\leq\lambda) \cdot \lambda^2 \leq (c_1 \lambda)\lambda^2.
\end{align*}
This proves our claim that \eqref{E:smaller} holds and finishes the proof.
\end{proof}

}

\begin{remark}\label{rmk:rsb-needed} Note that the proof above, when applied in the \emph{continuous setting}, shows that the reverse small ball property \eqref{eq:reverse-small-ball} is actually \emph{necessary} for declipping to hold stably for all $\lambda > 0$. 
Indeed, by repeating the proof above, we get that 
\begin{equation}\label{eq:limit-expect} 
\lim_{\varepsilon \to 0} \mathbb{E} \left( \frac{|\Phi_\lambda(\langle X_i,u\rangle) - \Phi_\lambda(\langle X_i,(1-\varepsilon)u \rangle)|^2}{\|\varepsilon u\|_2^2}\right) = \mathbb{E} |\langle X,u\rangle|^2\mathds{1}_{\{|\langle X,u\rangle| \le \lambda\}},
\end{equation} 
and hence if the left-hand side of \eqref{eq:limit-expect} is bounded from below by $c(\lambda),$ then we have that 
\[
\mathbb{P} \left(|\langle X,u\rangle| \le \lambda\right) \ge \frac{c(\lambda)}{\lambda^2},
\]
for each $u \in \S^{n-1}.$ This is \emph{exactly} the content of \eqref{eq:reverse-small-ball}, as desired. 

In a similar spirit, one readily sees that, in order for the lower bound stated in Theorem \ref{thm:main-vc} to hold, it must be the case that the lower bound in \eqref{E:proportional} holds as well. On the other hand, it is not a difficult task to check that, if one assumes that both inequalities in \eqref{E:proportional} hold, then the proof of Theorem \ref{thm:main-vc} above may be repeated almost verbatim, which shows that our proof actually applies to a wider class of distributions - a fact to which we alluded in the introduction. 

A completely characterization of all distributions such that the result of Theorem \ref{thm:main-vc} holds seems to be elusive: the examples provided by Remark \ref{rmk:ssb-sharp?} show that the upper bound in \eqref{E:proportional} is \emph{not} necessary in order for a conclusion similar to that of Theorem \ref{thm:main-vc} to hold. 

Furthermore, it is possible to show that the constant $c(\lambda)$ given by Proposition \ref{prop:double-sphere-ex} can be taken to be $c(\lambda) \sim \lambda^3$. This shows that a reasonable degree of \emph{singularity} of the distributions of the random vectors considered can be allowed so that a continuous version of Theorem \ref{thm:main-vc} still holds. The question of under which minimal conditions we may infer such a bound is an interesting one, which we wish to address in future work. 
\end{remark}

\subsection{Sharp unlimited sampling for uniformly spherically distributed vectors} We now deal with the results pertaining to sharp stable unlimited sampling.

{
\begin{proof}[Proof of Theorem \ref{T:folding}]
 Suppose that $X$ is a random vector which is uniformly distributed on the sphere $\sqrt{n}\S^{n-1}$ and $\lambda>0$. 
  We will consider separately the case that $u,v\in \B_{\R^n}$ with $\|u-v\|_2\leq \lambda$ and the case that $u,v\in \B_{\R^n}$ with $\|u-v\|_2\geq \lambda$. \medskip

\noindent \textbf{Case 1}:  $u,v\in \B_{\R^n}$ with $0<\|u-v\|_2\leq \lambda$. \medskip

We claim first that there exist constants  $C,c_1 > 0$ so that the following holds: If $m\ge C n$ and $(X_j)_{j=1}^m$ are uniformly distributed in $\sqrt{n}\S^{n-1}$ then with probability at least $1-e^{-c_1m}$, we have  that 
\begin{equation}\label{E:12}
  |\{j:1/2\|w\|_2\leq |\langle X_j,w\rangle|\leq\|w\|_2\}|\geq Cm\hspace{1cm}\textrm{ for all }w\in\R^n.  
\end{equation}
A way to see this is by considering, once more, the class $\mathcal{F} = \{\mathds{1}\left( 1/2 \le |\langle X_i,w\rangle| \le 2 \right) \colon w \in \S^{n-1}\}.$ Since elements in this class may be represented as unions of up to four hyperplanes, it follows that $vc(\mathcal{F}) \lesssim n.$ Hence, by applying Proposition \ref{prop:vc-quantif-lln} to this class, the assertion follows. We may therefore assume that \eqref{E:12} holds from now on.

A key difference between clipping and folding is that the operation $\mathcal{M}_\lambda$ satisfies that for any $s,t\in\R$ for which $|s-t|\leq \lambda$,  $|\mathcal{M}_\lambda(s)-\mathcal{M}_\lambda(t)|\geq |s-t|$. In particular, as $0<\|u-v\|_2\leq\lambda$ we have that
\begin{align*}
\frac{1}{m}\sum_{i=1}^m \frac{1}{\|u-v\|_{2}^2}\left|\mathcal{M}_{\lambda}(\langle X_i,u\rangle) - \mathcal{M}_{\lambda}(\langle X_i,v\rangle)\right|^2 &\ge \frac{1}{4m} \sum_{i=1}^m \mathds{1}\left(2^{-1}\le \left|\left\langle X_i,\frac{u-v}{\|u-v\|_2}\right\rangle\right|\le1\right)\\
&\ge \frac{C}{4} \hspace{1cm}\textrm{ by \eqref{E:12}.}
\end{align*}
This completes the proof in the case that $\|u-v\|_2\leq \lambda$. \medskip

\noindent \textbf{Case 2:}  $u,v\in \B_{\R^n}$ satisfy $\lambda\leq\|u-v\|_2$. \medskip

Note that this case implicitly assumes that $0<\lambda\leq 2$.
We now fix $1/4>a>0$ to be decided later.
Let $\Omega_a\subseteq\R$ be given by $\Omega_a=\cup_{m\in\Z}[(a+2m)\lambda,(1-a+2m)\lambda]$.  
Note that if $\langle X,u-v\rangle\in\Omega_a$ then
$\left|\mathcal{M}_{\lambda}(\langle X_i,u\rangle) - \mathcal{M}_{\lambda}(\langle X_i,v\rangle)\right|\geq a\lambda$.  Thus, we have that
\begin{equation}\label{E:fold_far}
\frac{1}{m}\sum_{i=1}^m \frac{1}{\|u-v\|_{2}^2}\left|\mathcal{M}_{\lambda}(\langle X_i,u\rangle) - \mathcal{M}_{\lambda}(\langle X_i,v\rangle)\right|^2 \ge \frac{(a\lambda)^2}{m} \sum_{i=1}^m \mathds{1}\left( \langle X_i,u-v\rangle\in\Omega_a\right).
\end{equation}
Notice that the right-hand side of \eqref{E:fold_far} only depends on $u-v$ and not on $u,v$ separately.  Thus, by substituting $w=u-v$, we may consider just the single vector $w\in 2\B_{\R^n}\setminus\lambda \B_{\R^n}$.  

For a subset $\Omega\subseteq\R$ we denote $|\Omega|$ to be the Lebesgue measure of $\Omega$.  Note that for every interval $I\subseteq\R$ with $|I|\geq\lambda$, we have that $|\Omega_a\cap I|\geq (1-2a)|I|$.  As $X$ is uniformly distributed in $\sqrt{n}\S^{n-1}$,  there exists a universal constant $c>0$ so that
 \begin{equation*}\label{E:mod_in}
 \mathbb{P}(\langle X,w\rangle\in\Omega_a)\geq 1-c\cdot a
\hspace{.5cm}\textrm{ for all $w\in2\R^n\setminus\lambda\R^n$.} 
\end{equation*}
Hence,
 \begin{equation*}\label{E:expected}
 \mathbb{E}\frac{1}{m}\sum_{i=1}^m \mathds{1}\left( \langle X_i,w\rangle\in\Omega_a\right)\geq 1-c \cdot a
\hspace{.5cm}\textrm{ for all $w\in2\R^n\setminus\lambda\R^n$.} 
\end{equation*}

By Proposition \ref{prop:bounded-differences} with $t=c\cdot a$,  for all $w\in 2\B_{\R^n}\setminus\lambda \B_{\R^n}$ we have with probability at least $1-e^{-Ca^2m}$ that 
\begin{equation}\label{E:big_p}
    \frac{1}{m}\sum_{i=1}^m \mathds{1}\left( \langle X_i,w\rangle\in\Omega_{2a}\right)\geq 1-c \cdot a.
\end{equation}

Let $\mathcal{N} = \{w_j\}_j$ be an $\vp$-net in  $2\B_{\R^n}$ with $|\mathcal{N}|\leq (6/\vp)^n$, where we set $\vp = c \cdot \lambda$ for some sufficiently small absolute constant $c > 0$. Then, we have with probability at least $1-e^{(c-\log(\lambda))n-C a^2m}$ that \eqref{E:big_p} holds for all $w \in \mathcal{N}.$ 

We now consider the task of extending this to all of $2\B_{\R^n}\setminus\lambda \B_{\R^n}$. In order to do so, we need to invoke once more Proposition \ref{prop:vershynin-subgauss} which guarantees that, with high probability, $(\frac{1}{\sqrt{m}}X_i)_{i=1}^m$ has upper frame bound $2$, or, in other words, \eqref{eq:upper-frame-bound} holds.  

This gives that there are uniform constants $c,C>0$ so that if $m\geq C\log(1/\lambda)n$ then with probability at least $1-e^{-cm}$, \eqref{eq:upper-frame-bound} and  \eqref{E:big_p} hold for all $w \in \mathcal{N}.$ We may hence assume that this is the case.

Let $w\in 2\B_{\R^n}\setminus\lambda\B_{\R^n}$ and choose $w_j \in \mathcal{N}$ such that $\|w-w_j\|_2\leq\vp$.  We have, in a similar way to the proof of Theorem \ref{thm:main-vc}, 
\begin{align*}
 \frac{1}{m}\sum_{i=1}^m \mathds{1}\left( \langle X_i,w\rangle\in\Omega_{a}\right)&\ge  \frac{1}{m}\sum_{i=1}^m \mathds{1}\left( \langle X_i,w_j\rangle\in\Omega_{2a}\textrm{ and }|\langle X_i,w-w_j\rangle|\leq a\lambda\right)\\
 &\ge  \frac{1}{m}\sum_{i=1}^m \mathds{1}\left( \langle X_i,w_j\rangle\in\Omega_{2a})\right)- \frac{1}{m}\sum_{i=1}^m \mathds{1}\left(|\langle X_i,w-w_j\rangle|> a\lambda\right)\\
  &\ge  (1-c\cdot a)- (a\lambda)^{-2}\frac{1}{m}\sum_{i=1}^m |\langle X_i,w-w_j\rangle|^2\\
    &\ge  (1-c\cdot a)- (a\lambda)^{-2} 2\|w-w_i\|_2^2  \ge c.
\end{align*}

By \eqref{E:fold_far}, this gives for all $u,v\in \B_{\R^n}$ with $\|u-v\|_2\geq\lambda$ that 
\begin{align*}
\frac{1}{m}\sum_{i=1}^m \frac{1}{\|u-v\|_{2}^2}\left|\mathcal{M}_{\lambda}(\langle X_i,u\rangle) - \mathcal{M}_{\lambda}(\langle X_i,v\rangle)\right|^2 &\ge \frac{(a\lambda)^2}{m} \sum_{i=1}^m \mathds{1}\left( \langle X_i,u-v\rangle\in\Omega_a\right)  
\geq c \cdot \lambda^2.
\end{align*}
This finishes the proof. 
\end{proof}

}

{  

\begin{remark} From the proof above, one may be tempted to conjecture that there should actually be \emph{no} dependence on $\lambda$ \emph{at all} in the number of samples: indeed, we may, instead of the proof above, use the lower bound
\begin{align*}
 \frac{1}{m}\sum_{i=1}^m \mathds{1}\left( \langle X_i,w\rangle\in\Omega_{a}\right)&\ge  \frac{1}{m}\sum_{i=1}^m \mathds{1}\left( \langle X_i,w_j\rangle\in\Omega_{2a}\textrm{ and }|\mathcal{M}_{\lambda}\left(\langle X_i,w-w_j\rangle|\right)|\leq a\lambda\right)\\
 &\ge  \frac{1}{m}\sum_{i=1}^m \mathds{1}\left( \langle X_i,w_j\rangle\in\Omega_{2a})\right)- (a\lambda)^{-2}\frac{1}{m}\sum_{i=1}^m (a\lambda)^2\mathds{1}\left(|\mathcal{M}_{\lambda}\left( \langle X_i,w-w_j\rangle\right)|> a\lambda\right). 
\end{align*}
So, one is naturally led to consider pointwise estimates for quantities such as 
\[
\frac{1}{m} \sum_{j=1}^m \mathds{1} \left( |\mathcal{M}_{\lambda}\left( \langle X_i,v\rangle \right)| > t\right). 
\]
On the other hand, a direct argument shows that we have that 
\[
\mathbb{P}( |\mathcal{M}_\lambda ( \langle X,v\rangle)| > a\cdot \lambda ) \to p(a),
\]
whenever $v \in \R^n$ is fixed, where $p(a)$ converges to $1$ as $a \to 0$. Hence, one \emph{cannot} leverage any smallness in norm of $v$ in order to conclude bounds on the probability above, making it unclear how to execute the previously discussed idea with the current methods. 
\end{remark}}
\begin{remark}\label{rmk:sharpness} In analogy to what was done above in the proof of Theorem \ref{thm:best-lam}, it is also possible to prove that the dependency on $\lambda$ in Theorem \ref{T:folding} is \emph{sharp}. 

Indeed, take $v = 0$ and let $U$ be uniformly distributed on $\S^{n-1}$. We have that 
\[
\mathbb{E} \, \frac{1}{m \|u-v\|_2^2} \sum_{i=1}^m |\mathcal{M}_\lambda(\langle x_i,U\rangle) - \mathcal{M}_\lambda(\langle x_i,v\rangle)|^2 = \mathbb{E} |\mathcal{M}_\lambda(\langle x,U\rangle)|^2,
\]
where $x \in \sqrt{n} \S^{n-1}$ is arbitrary. By denoting by $d \mu_x$ the pushforward of the uniform measure on $\S^{n-1}$ by $\langle x, \cdot \rangle$; that is, the one-dimensional measure such that 
\[
\mathbb{E} f(\langle x,U\rangle) = \int_{\R} f(s) \, d\mu_x(s),
\]
we have that 
\[
\frac{1}{\lambda^2} \mathbb{E} |\mathcal{M}_\lambda(\langle x,U\rangle)|^2 = \int_\R \frac{\left| \mathcal{M}_\lambda (s)\right|^2}{\lambda^2} \, d\mu_{x}(s).
\]
Now, the simplest version of the Funk-Hecke formula \cite[Theorem~1.2.9]{dai} says that the measure $d\mu_x$ is \emph{independent of }$x \in \sqrt{n}\S^{n-1}$, and, as a matter of fact, we have 
\[
d\mu_x(s) = \frac{\Gamma\left( \frac{n}{2}\right)}{\sqrt{n} \Gamma\left( \frac{n-1}{2}\right)} \cdot (1-s^2)^{\frac{n-3}{2}} \cdot \mathds{1}_{[-1,1]}(s). 
\]
Since, on the other hand, we also have that 
\[
\frac{\left| \mathcal{M}_\lambda (s)\right|^2}{\lambda^2} = |\mathcal{M}_1(s/\lambda)|^2,
\]
a change of variables shows that 
\begin{equation}\label{eq:equality-weak-conv} 
\int_\R \frac{\left| \mathcal{M}_\lambda (s)\right|^2}{\lambda^2} \, d\mu_{x}(s) = \frac{\Gamma\left( \frac{n}{2}\right)}{\sqrt{n} \Gamma\left( \frac{n-1}{2}\right)}  \cdot \lambda \int_{-1/\lambda}^{1/\lambda} |\mathcal{M}_1(t)|^2 \cdot  (1-\lambda^2 t^2)^{\frac{n-3}{2}}\, dt,
\end{equation}
and it is then not difficult to check that the right-hand side of \eqref{eq:equality-weak-conv} is bounded from below by a dimensionless, positive constant as $\lambda \to 0$. In other words, the result of Theorem \ref{T:folding} cannot be improved upon for general $u,v \in \B_{\R^n}$, as we wanted to show. 
\end{remark}

\subsection{Sharp declipping in the sparse setting} {We now move on to the proof of Theorem \ref{thm:main-sparse}. As we shall see below, the proof of that result may be obtained by essentially running again the arguments used in the proof of Theorem \ref{thm:main-vc}.
\begin{proof}[Proof of Theorem \ref{thm:main-sparse}] We now highlight the changes needed in the proof of Theorem \ref{thm:main-vc} in order to obtain Theorem \ref{thm:main-sparse}. 

First, note that the set $T_s$ of $s$-sparse vectors can be covered by an $\varepsilon$-net of size $O\left( {n \choose s} \cdot \left( \frac{C}{\varepsilon}\right)^{s}\right)$, and hence, repeating the proof of Theorem \ref{thm:main-vc}, we get that, for $
m \ge C \left( s\log(1/\lambda) + s \log\left(\frac{en}{s}\right)\right),$ the exact same assertions as in that result hold. This proves the first statement in Corollary \ref{cor:sparse-rec}.

In order to prove the second assertion in Theorem \ref{thm:main-sparse}, we need the following result -- for its proof, we refer the reader to \cite[Lemma~3.4]{plan2013one}.
\begin{lemma}[Lemma 3.4 in \cite{plan2013one}]\label{lemma:epsilon-effect-sparse} There exists an $\varepsilon$-net $\mathcal{N}_{s,n}$ of $\Sigma_s \cap \S^{n-1}$ of size 
\[
\log|\mathcal{N}_{s,n}| \le \frac{C}{\varepsilon^2} s \log\left( \frac{en}{s}\right). 
\]
\end{lemma}

We then note the following improvement over the proof of Theorem \ref{thm:spherical}: we may use Proposition \ref{prop:vc-quantif-lln} together with Proposition \ref{prop:bounded-differences} (in the same way as in the proof of Theorem \ref{General distributions}, for instance) in order to say that, with probability at least $1 - e^{-c \lambda^2 m}$, we have 
\[
\frac{1}{m} \sum_{i=1}^m \mathds{1} \left( |\langle X_i, u-u_0\rangle| \ge \lambda\right) \le \mathbb{P} \left( |\langle X,u-u_0\rangle| \ge \lambda\right) + C \sqrt{\frac{n}{m}} + t.
\]
We now pick $t = c\lambda, \, m \gtrsim \lambda^{-2} \cdot n$, and choose $\varepsilon = c \cdot \frac{\lambda}{\sqrt{\log\left( \frac{1}{\lambda}\right)}}$ for the $\varepsilon$-net given by Lemma \ref{lemma:epsilon-effect-sparse} above. This shows that 
\[
\mathbb{P}\left( |\langle X,u-u_0\rangle| \ge \lambda\right) \le C \cdot e^{- \frac{\lambda^2}{\|u-u_0\|_2^2}} \le c \cdot \lambda,
\]
where we used the fact that $X$ is a subgaussian vector. It follows then that, with probability at least $1 - e^{-c \lambda^2 m},$ we have 
\[
\frac{1}{m} \sum_{i=1}^m \mathds{1}{ \left( |\langle X,u-u_0\rangle| \ge \lambda\right)} \le c \cdot \lambda.
\]
By using this together with the bound on the sizes of $\varepsilon$-nets (Lemma \ref{lemma:epsilon-effect-sparse}) in the set of effectively sparse vectors on the sphere $\Sigma_s \cap \S^{n-1}$ in the proof of Theorem \ref{thm:spherical}, we obtain that, for $m \ge C \lambda^{-3} \log\left(1/\lambda\right) s \log\left( \frac{en}{s}\right)$, our result holds as claimed. 
\end{proof}

 \begin{remark} We note that the dependency on $\lambda$ in the lower bound in Theorem \ref{thm:main-sparse} is again \emph{best possible}. Analogously, by repeating the proof of Theorem \ref{thm:best-lam}, we obtain that, once more, the bound asserted by Theorem \ref{thm:main-sparse} is also best possible when considering any sequence of vectors $x_1,\dots,x_m$ on $\sqrt{n} \S^{n-1}$, with $m \ge n.$

\end{remark}
}

\section{Sharp declipping for spherical vectors and one-bit compressed sensing}\label{sec:one-bit}

We start with a couple of lemmas, the first of which shows that, whenever two vectors on the unit sphere are at least separated by a bit, we can get a good lower bound on the probability that one of the vectors gets clipped above a certain threshold and the other below.

\begin{lemma}\label{lemma:estimates-sets} There is an absolute constant $C>0$ for which the following assertion holds: Let $\lambda >0$ and $u,v \in \S^{n-1}$ satisfy that $\|u-v\|_2 \ge C \lambda.$ Then there is an absolute constant $c>0$ such that 
\[
\mathbb{P} \left( \langle X,u\rangle \le -\lambda, \, \langle X, v \rangle \ge \lambda \right) \ge c \cdot \|u-v\|_2.
\]
\end{lemma}
\begin{proof} We start by noting that we may write 
\[
X = \sqrt{n} \frac{Y}{\|Y\|_2} \, , 
\]
where $Y \distas{d} \mathcal{N}(0,I_{n \times n})$. Hence, for each $C'>0$,
\[
\mathbb{P} \left( \langle X,u\rangle \le -\lambda, \, \langle X, v \rangle \ge \lambda \right)  = \mathbb{P} \left( \langle Y,u\rangle \le - \lambda \frac{\|Y\|_2}{\sqrt{n}}, \, \langle Y, v \rangle \ge \lambda \frac{\|Y\|_2}{\sqrt{n}}\right)  
\]
\[
\ge \mathbb{P} \left( \langle Y,u\rangle \le - \lambda \frac{\|Y\|_2}{\sqrt{n}}, \, \langle Y, v \rangle \ge \lambda \frac{\|Y\|_2}{\sqrt{n}}, \, \frac{1}{C'} \sqrt{n} \le \|Y\|_2 \le C' \sqrt{n}\right)  
\]
\[
\ge \mathbb{P} \left( \langle Y,u\rangle \le -C' \lambda, \langle Y,v\rangle \ge C' \lambda, \frac{1}{C'} \sqrt{n} \le \|Y\|_2 \le C' \sqrt{n}\right)  
\]
\begin{equation}\label{eq:two-probs}
\ge \mathbb{P}\left(  \langle Y,u\rangle \le -C' \lambda, \langle Y,v\rangle \ge C' \lambda\right) - \mathbb{P}\left( \langle Y,u\rangle \le 0, \, \langle Y,v \rangle \ge 0, \, \frac{1}{C'} \sqrt{n} > \|Y\|_2 \text{ or } \|Y\|_2 > C' \sqrt{n}\right).
\end{equation}
We need to find a suitable lower bound for the first term and an upper bound for the second one. Let us thus start with the latter. Since $Y$ is invariant under rotations and since the event 
$$\{ \langle Y,u \rangle \le 0, \langle Y,v \rangle \ge 0\}$$
denotes an angle in $\R^n$ with opening of the order of $\|u-v\|_2$, we have that the latter probability is at most 
\[
\|u-v\|_2 \cdot \mathbb{P}\left( \|Y\|_2 \ge C' \sqrt{n} \text{ or } \|Y\|_2 \le \frac{1}{C'} \sqrt{n}\right) \le \|u-v\|_2 \cdot e^{-(C')^2}.
\]
We now seek a lower bound on the other probability. Since the vector $Y$ is rotation-invariant, we may write the probability we wish to estimate explicitly as 
\[
\int_{C'\lambda}^{\infty} \int_{-\infty}^{-C'\lambda} \frac{1}{2\pi\sqrt{1-\rho^2}} \exp\left\{ - \frac{1}{2(1-\rho^2)} (x^2 + y^2 - 2\rho xy)\right\}  \, dx \, dy,
\]
where we let 
$$\rho = 1 - \frac{\|u-v\|_2^2}{2}.$$
Changing variables $x = \|u-v\|_2 \cdot x', \, y = -\|u-v\|_2\cdot y',$ we see that the integral above is bounded from below by 
\begin{align}\label{eq:gaussian-two-d}
\|u-v\|_2 \cdot & \int_{C' \frac{\lambda}{\|u-v\|_2}}^{\infty} \int_{C' \frac{\lambda}{\|u-v\|_2}}^{\infty}  \exp\left\{ - \left( (x')^2 + (y')^2 + 2\rho x' y'\right)\right\} \, d x' \, dy' \cr 
 & \ge \|u-v\|_2 \int_{C'/100}^{\infty} \int_{C'/100}^{\infty} \exp(-2((x')^2 + (y')^2)) \, dx' \, dy'.
\end{align}
Now, since 
\[
\int_x^{\infty} e^{-2t^2} \, dt = \frac{1}{\sqrt{2}} \int_{\sqrt{2}x}^{\infty} e^{-s^2} \, ds \ge \frac{1}{4x} e^{-2x^2},
\] 
we have that 
\[
\int_{C'/100}^{\infty} e^{-2 (x')^2} \, dx' \ge \frac{1}{C'} e^{-\frac{(C')^2}{5000}}.
\]
It follows that the right-hand side of \eqref{eq:gaussian-two-d} is bounded from below by 
$$e^{- \frac{(C')^2}{1000}} \|u-v\|_2,$$ 
as long as $\|u-v\|_2 > 100 \lambda,$ which shows that the first probability in \eqref{eq:two-probs} above is larger than two times the second, finishing our proof, as long as we take $C \ge 100$ in the statement of the lemma. 
\end{proof}

The next result complements the previous one, showing that, on the other hand, even if both vectors are close by, then the probability of both of them \emph{not} getting clipped is non-trivial. 

\begin{lemma} For any $C>0$ there is $c>0$ such that the following holds. Let $u,v \in \S^{n-1}$ be such that $\|u-v\|_2 \le C \lambda.$ Then, for $\theta > 0$ a sufficiently small absolute constant, we have
\[
\mathbb{P}\left( \left| \left\langle X,\frac{u-v}{\|u-v\|_2}\right\rangle \right| \ge \theta \, \cap |\langle X,u\rangle| \le \lambda \, \cap |\langle X,v \rangle| \le \lambda\right) \ge c \cdot \lambda. 
\]
\end{lemma}

\begin{proof} We do the same as in the previous lemma: we have that 
\begin{align*} 
& \mathbb{P}\left( \left| \left\langle X,\frac{u-v}{\|u-v\|_2}\right\rangle \right| \ge \theta,|\langle X,u\rangle| \le \lambda, |\langle X,v \rangle| \le \lambda\right) \cr 
& \ge \mathbb{P} \left( \left| \left\langle Y,\frac{u-v}{\|u-v\|_2}\right\rangle \right| \ge C'\theta,|\langle Y,u\rangle| \le \frac{\lambda}{C'}, |\langle Y,v\rangle| \le \frac{\lambda}{C'}, \, \frac{1}{C'} \sqrt{n} \le \|Y\|_2 \le C' \sqrt{n}\right). 
\end{align*} 
Now note that, by rotation-invariance of $Y$ once more, we may assume that $u = e_1$ and $v = \alpha e_1 + \beta e_2$, where $\{e_j\}_{j=1}^n$ denotes the canonical basis in $\R^n$. Since we have $\|u-v\|_2 \le C \lambda$, it follows that 
\[
(1-\alpha)^2 + \beta^2 \le C^2 \lambda^2. 
\]
In particular, an elementary calculation shows that, under the conditions of the lemma, we have that belonging to the event 
\[
\left\{ |Y_1| \le \frac{\lambda}{\tilde{C}}, 2C'\theta \le |Y_2|\le \frac{1}{\tilde{C}},\frac{\sqrt{n}}{C'} \le Y_3^2 + \cdots + Y_n^2 \le C'\sqrt{n} - 2 \right\}
\]
implies belonging to the event
\[
\left\{ \left| \left\langle Y,\frac{u-v}{\|u-v\|_2}\right\rangle \right| \ge C'\theta,|\langle Y,u\rangle| \le \frac{\lambda}{C'}, |\langle Y,v\rangle| \le \frac{\lambda}{C'}, \, \frac{1}{C'} \sqrt{n} \le \|Y\|_2 \le C' \sqrt{n}\right\}.
\]
Here, note that the absolute constant $\tilde{C}$ depends only on the constant $C>0.$ Hence, we have 
\[
\mathbb{P} \left( \left| \left\langle Y,\frac{u-v}{\|u-v\|_2}\right\rangle \right| \ge C'\theta,|\langle Y,u\rangle| \le \frac{\lambda}{C'}, |\langle Y,v\rangle| \le \frac{\lambda}{C'}, \, \frac{1}{C'} \sqrt{n} \le \|Y\|_2 \le C' \sqrt{n}\right) 
\]
\[
\ge \mathbb{P}\left( |Y_1| \le \frac{\lambda}{20C'}\right) \cdot \mathbb{P}\left( 2C'\theta \le |Y_2| \le \frac{1}{20}\right) \cdot \mathbb{P}\left(\frac{\sqrt{n}}{C'} \le Y_3^2 + \cdots + Y_n^2 \le C'\sqrt{n} - 2 \right),
\]
by independence of the entries of $Y$. Since the two events
\[
\left\{ \frac{\sqrt{n}}{C'} \le Y_3^2 + \cdots + Y_n^2 \le C'\sqrt{n} - 2 \right\}, \quad \left\{ 2C'\theta \le |Y_2| \le \frac{1}{20}\right\}
\]
have probability bounded from below by an absolute constant as long as we take $\theta$ to be a sufficiently small absolute constant and since 
\[
\mathbb{P}\left( |Y_1| \le \frac{\lambda}{20C'}\right) \ge c \cdot \lambda, 
\]
for some small absolute constant $c>0,$ we have that 
\[
\mathbb{P}\left( \left| \left\langle X, \frac{u-v}{\|u-v\|_2} \right\rangle \right| > \theta,|\langle X,u\rangle| \le \lambda, |\langle X,v \rangle| \le \lambda\right) \ge c \cdot \lambda,
\]
for some possibly different absolute constant $c>0,$ as desired. 
\end{proof} 

We now move to the proof of Theorem \ref{thm:spherical}. 
\begin{proof}[Proof of Theorem \ref{thm:spherical}] We divide our task once more into two cases. \medskip

\noindent\textbf{Case 1: $\|u-v\|_2 \ge C\lambda$.} In this case, we bound 
\[
\frac{1}{m} \sum_{i=1}^m \frac{|\Phi_\lambda(\langle X_i,u\rangle) - \Phi_{\lambda}(\langle X_i,v\rangle)|^2}{\|u-v\|_2^2} \ge \frac{4\lambda^2}{\|u-v\|_2^2} \frac{1}{m} \sum_{i=1}^m \mathds{1}\left( \langle X_i,u\rangle \le - \lambda, \langle X_i, v \rangle \ge \lambda \right).  
\]
Now, take $\mathcal{N}'$ to be a $c\lambda^2$-net on $\B_{\R^n}$. Let $u_0,v_0$ denote the closest elements of $\mathcal{N}'$ to $u,v$, respectively. We then bound 
\begin{align*}
\frac{1}{m} & \sum_{i=1}^m \mathds{1}( \langle X_i,u\rangle \le -\lambda, \langle X_i,v\rangle \ge \lambda) \cr 
&\ge \frac{1}{m} \sum_{i=1}^m \mathds{1}( \langle X_i,u_0\rangle \le -2\lambda, \langle X_i,v_0\rangle \ge 2\lambda, |\langle X_i,u-u_0\rangle| \le \lambda,|\langle X_i,v-v_0\rangle|\le \lambda) \cr 
& \ge \frac{1}{m} \left( \sum_{i=1}^m \mathds{1}( \langle X_i,u_0\rangle \le -2\lambda, \langle X_i,v_0\rangle \ge 2\lambda) - \sum_{i=1}^m \mathds{1}(|\langle X_i,u-u_0\rangle| \ge \lambda \cup |\langle X_i,v-v_0\rangle|\ge \lambda)\right).
\end{align*}
Since we know that, for each individual pair $(u_0,v_0)\in \mathcal{N}' \times \mathcal{N}'$ such that $\|u_0 - v_0\| \ge \frac{C}{2}\lambda,$ we have with probability at least $1 - e^{-c \lambda m}$ that 
\[
\frac{1}{m} \sum_{i=1}^m \mathds{1}( \langle X_i,u_0\rangle \le -2\lambda, \langle X_i,v_0\rangle \ge 2\lambda) \ge \frac{c}{2} \|u_0-v_0\|_2.
\]
We next note that
\begin{align*} 
\frac{1}{m}&\sum_{i=1}^m \mathds{1}(|\langle X_i,u-u_0\rangle| \ge \lambda \cup |\langle X_i,v-v_0\rangle|\ge \lambda) \cr 
& \le \frac{1}{m} \sum_{i=1}^m \mathds{1}(|\langle X_i,u-u_0\rangle| \ge \lambda ) + \frac{1}{m} \sum_{i=1}^m \mathds{1}(|\langle X_i,v-v_0\rangle|\ge \lambda) \cr 
& \le \frac{1}{m\lambda^2} \sum_{i=1}^m \left( |\langle X_i,u-u_0\rangle|^2 + |\langle X_i,v-v_0\rangle|^2\right) \cr 
& \le C\lambda^2,
\end{align*} 
where the last line follows from Proposition \ref{prop:vershynin-subgauss}, and it holds with probability at least $1- e^{-cm}.$ Thus, with probability at least $1-|\mathcal{N}'|^2e^{-c \lambda m} \ge 1 - e^{C n \log(1/\lambda) - c \lambda m}$, we have that 
\[
\frac{1}{m} \sum_{i=1}^m \mathds{1}(\langle X_i,u\rangle \le -\lambda, \langle X_i,v\rangle \ge \lambda) \ge \frac{c}{3} \|u-v\|_2,
\]
which concludes the claim in this case, since $2 \ge \|u-v\|_2 \ge C\lambda$ and $m \gtrsim \lambda^{-1} \log(1/\lambda)n.$ \medskip

\noindent\textbf{Case 2: $\|u-v\|_2 \le C\lambda$.} In this case, one cannot do better than restricting matters to the clipped vectors. Indeed, we use the bound 
\[
\frac{1}{m} \sum_{i=1}^m \frac{|\Phi_\lambda(\langle X_i,u\rangle) - \Phi_{\lambda}(\langle X_i,v\rangle)|^2}{\|u-v\|_2^2} \ge \frac{1}{m} \sum_{i=1}^m \left| \left\langle X_i,\frac{u-v}{\|u-v\|_2}\right\rangle\right|^2 \mathds{1}\left(|\langle X_i,u\rangle| \le \lambda, |\langle X_i,v\rangle| \le \lambda\right)
\]
\[
\ge \frac{\theta^2}{m} \sum_{i=1}^m \mathds{1} \left( \left| \left\langle X_i,\frac{u-v}{\|u-v\|_2}\right\rangle\right| > \theta, |\langle X_i,u\rangle| \le \lambda, |\langle X_i,v\rangle| \le \lambda\right).
\]
By arguing as in the previous case, we are able to conclude that, with probability at least $1 - e^{-c \lambda m}$, we have 
\[
\frac{1}{m} \sum_{i=1}^m \mathds{1} \left( \left| \left\langle X_i,\frac{u-v}{\|u-v\|_2}\right\rangle\right| > \theta, |\langle X_i,u\rangle| \le \lambda, |\langle X_i,v\rangle| \le \lambda\right) \ge \frac{c}{3} \lambda, \, \forall \, u,v\in\S^{n-1}. 
\]
This finishes the proof. 
\end{proof}

\begin{proof}[Proof of Corollary \ref{cor:sparse-rec}] The proof of Corollary \ref{cor:sparse-rec} may be achieved by simply combining the proof of Theorem \ref{thm:spherical} with the proof of Theorem \ref{thm:main-sparse} sketched above. We omit the details. 
\end{proof}

\begin{remark}\label{remark on samples} Note that, in the statement of Theorem \ref{cor:sparse-rec}, the minimal asserted number of samples required for the second part of the result to hold depends at least like $\lambda^{-3} \log\left( \frac{1}{\lambda}\right)$ times the square of the Gaussian width of the set $\Sigma_s$ of effectively sparse vectors. We note that, in an analogous manner, it has been conjectured in \cite{plan-vershynin-2} that a dependence of $\delta^{-2}$ times the square of the Gaussian width should be necessary -- and sufficient -- in order for a random hyperplane tesselation on the set to approximate the distance between two points up to confidence $\delta$ with high probability. 

In that regard, we highlight \cite{jacques2013robust}, where in the case of the \emph{exactly} sparse vectors a similar result has been obtained as Corollary \ref{cor:one-bit-cs}, and the work of S.~Dirksen and S.~Mendelson \cite{dirksen2021non} which showed that, under the effect of adding noise before quantization -- the so-called process of \emph{dithering} -- one may actually recover effectively sparse vectors from the whole unit \emph{ball} with relatively few samples. In this context, a natural step would be to consider a suitable dithering version of the problems above, which we shall delegate to future work.
\end{remark}
Finally, we present a proof of Corollary \ref{cor:one-bit-cs}.  
\begin{proof}[Proof of Corollary \ref{cor:one-bit-cs}] 
We focus on the first assertion of Corollary \ref{cor:one-bit-cs}. Note that the proof of \textbf{Case 1} of Theorem \ref{thm:spherical} applies with the modifications highlighted for Corollary \ref{cor:sparse-rec} above if we only suppose that $\|u-v\|_2 \ge \delta$ and we redo the argument accordingly. In more specific words, we can bound
\[
\frac{1}{m} \sum_{i=1}^m \frac{|\Phi_\lambda(\langle X_i,u\rangle) - \Phi_{\lambda}(\langle X_i,v\rangle)|^2}{\|u-v\|_2^2} \ge \frac{4\lambda^2}{\|u-v\|_2^2} \frac{1}{m} \sum_{i=1}^m \mathds{1}\left( \langle X_i,u\rangle \le - \lambda, \langle X_i, v \rangle \ge \lambda \right).  
\]
We may still estimate accordingly in this case by 
\begin{align*}
\frac{1}{m} & \sum_{i=1}^m \mathds{1}( \langle X_i,u\rangle \le -\lambda, \langle X_i,v\rangle \ge \lambda) \cr 
&\ge \frac{1}{m} \sum_{i=1}^m \mathds{1}( \langle X_i,u_0\rangle \le -\lambda-\delta, \langle X_i,v_0\rangle \ge \lambda + \delta, |\langle X_i,u-u_0\rangle| \le \delta,|\langle X_i,v-v_0\rangle|\le \delta) \cr 
& \ge  \left( \frac{1}{m}\sum_{i=1}^m \mathds{1}( \langle X_i,u_0\rangle \le -\lambda-\delta, \langle X_i,v_0\rangle \ge \lambda+\delta) \right. \cr 
& \left. - \frac{1}{m}\sum_{i=1}^m \mathds{1}(|\langle X_i,u-u_0\rangle| \ge \delta \cup |\langle X_i,v-v_0\rangle|\ge \delta)\right).
\end{align*}
As long as we take $\lambda < \delta,$ the first sum is bounded from below by 
\[
\frac{1}{m} \sum_{i=1}^m \mathds{1}( \langle X_i,u_0\rangle \le -2\delta, \langle X_i,v_0\rangle \ge 2\delta),
\]
which like in the proof of Theorem \ref{thm:spherical}, is bounded from below by $c \|u-v\|_2$ in an event of probability at least $1 - |\mathcal{N}|^2 e^{- c\delta m}$, while the second sum is bounded from above by $c \cdot \delta$ with probability at least $1 - e^{-c \delta^2 m}$. By choosing $\mathcal{N}$ to be a $c\frac{\delta}{\sqrt{\log\left(\frac{1}{\delta}\right)}}$-net of $\Sigma_s \cap \S^{n-1}$, we obtain that, with probability at least $1 - e^{-c \delta^2 m}$, we have that for all $u, v \in \Sigma_s \cap \S^{n-1}$with $\|u-v\|_2 \ge C \delta,$ 
\[
\frac{1}{m} \sum_{i=1}^m \frac{|\Phi_\lambda(\langle X_i,u\rangle) - \Phi_{\lambda}(\langle X_i,v\rangle)|^2}{\|u-v\|_2^2} \ge c \cdot \lambda^2,
\]
whenever $m \ge C \delta^{-3} \log\left( \frac{1}{\delta}\right) s \log\left( \frac{en}{s}\right).$ A quick analysis of the proof above shows that the events above can be seen to be \emph{independent} of $\lambda < \delta$. This concludes the first assertion of Corollary \ref{cor:one-bit-cs}. 

\vspace{2mm}

In order to prove the second assertion of Corollary \ref{cor:one-bit-cs}, let $(u_{\delta},v_{\delta}) \in (\Sigma_s\cap\S^{n-1})^2$ be such that $\|u_\delta - v_\delta\|_2 \ge \delta$ and 
\[
\inf_{\substack{{u,v \in \Sigma_s \cap \S^{n-1}} \\ {\|u-v\|_2 \ge \delta}}} \frac{1}{m\|u-v\|_2^2} \sum_{i=1}^m |\text{sign}(\langle X_i, u \rangle) - \text{sign}(\langle X_i,v\rangle)|^2 \]
\[
\ge \frac{1}{2} \frac{1}{m\|u_\delta-v_\delta\|_2^2} \sum_{i=1}^m |\text{sign}(\langle X_i, u_\delta \rangle) - \text{sign}(\langle X_i,v_\delta\rangle)|^2
\]
\[
= \frac{1}{2m\|u_\delta - v_\delta\|_2^2} \lim_{\lambda \to 0} \sum_{i=1}^m \frac{|\Phi_\lambda(\langle X_i,u_\delta\rangle) - \Phi_\lambda(\langle X_i,v_{\delta}\rangle)|^2}{\lambda^2}.
\]
By the first assertion, this last right-hand side is bounded from below by an absolute constant with probability at least $1 - e^{-c\delta^2 m}$, which concludes the proof of the second assertion, and hence also that of Corollary \ref{cor:one-bit-cs}. 
\end{proof}

\printbibliography[heading=bibintoc,title=References]

\end{document}